\documentclass[journal]{IEEEtran}
\usepackage{amsmath,amsfonts}
\usepackage{algorithmic}
\usepackage{array}
\usepackage[caption=false,font=normalsize,labelfont=sf,textfont=sf]{subfig}
\usepackage{textcomp}
\usepackage{stfloats}
\usepackage{url}
\usepackage{verbatim}
\usepackage{graphicx}
\usepackage{amsmath}
\usepackage{cite}
\usepackage{float}
\usepackage{mathrsfs} 
\usepackage{amssymb}
\usepackage{siunitx}
\usepackage[linesnumbered,ruled,vlined]{algorithm2e}
\usepackage[colorlinks=true,linkcolor=blue,citecolor=blue,urlcolor=blue]{hyperref}
\newtheorem{definition}{Definition}
\newtheorem{theorem}{Theorem}
\newtheorem{lemma}{Lemma}
\newtheorem{proof}{Proof}
\newtheorem{assumption}{Assumption}

\newtheorem{proposition}{Proposition}

\begin{document}
\title{Risk Assessment for Nonlinear Cyber-Physical Systems under Stealth Attacks}
\author{Guang Chen, Zhicong Sun, Yulong Ding, Shuang-hua Yang

\thanks{This work was supported in part by the National Natural Science Foundation of China under Grant 92067109, Grant 61873119, and Grant 62211530106; in part by the Shenzhen Science and Technology Program under Grant ZDSYS20210623092007023 and Grant GJHZ20210705141808024; and in part by the Educational Commission of Guangdong Province under Grant 2019KZDZX1018. (Corresponding author: Shuang-Hua Yang.)}

\thanks{G. Chen is with a degree from the Department of Mechanical and Energy Engineering at the Southern University of Science and Technology, China (email: 12032433@mail.sustech.edu.cn).}

\thanks{Z. Sun is with the Department of Civil and Environmental Engineering, Hong Kong Polytechnic University, Hong Kong SAR, China, and also with the Institute of Advanced Computing and Digital Engineering, Shenzhen Institute of Advanced Technology, Chinese Academy of Sciences, Shenzhen, China (email: zhicong.sun@connect.polyu.hk).}

\thanks{Y. Ding is with the Shenzhen Key Laboratory of Safety and Security for Next Generation of Industrial Internet, SUSTech, Shenzhen 518055, China, and also with the Department of Computer Science and Engineering, SUSTech, Shenzhen 518055, China (email: dingyl@sustech.edu.cn).}

\thanks{S. Yang is with the Shenzhen Key Laboratory of Safety and Security for Next Generation of Industrial Internet, SUSTech, Shenzhen 518055, China, and also with the Department of Computer Science, University of Reading, UK (email: yangsh@sustech.edu.cn).}}

\markboth{Journal of \LaTeX\ Class Files,~Vol.~14, No.~8, August~2021}
{Shell \MakeLowercase{\textit{et al.}}: A Sample Article Using IEEEtran.cls for IEEE Journals}

\IEEEpubid{0000--0000/00\$00.00~\copyright~2021 IEEE}

\maketitle

\begin{abstract}
Stealth attacks pose potential risks to cyber-physical systems because they are difficult to detect. Assessing the risk of systems under stealth attacks remains an open challenge, especially in nonlinear systems. To comprehensively quantify these risks, we propose a framework that considers both the reachability of a system and the risk distribution of a scenario. We propose an algorithm to approximate the reachability of a nonlinear system under stealth attacks with a union of standard sets. Meanwhile, we present a method to construct a risk field to formally describe the risk distribution in a given scenario. The intersection relationships of system reachability and risk regions in the risk field indicate that attackers can cause corresponding risks without being detected. Based on this, we introduce a metric to dynamically quantify the risk. Compared to traditional methods, our framework predicts the risk value in an explainable way and provides early warnings for safety control. We demonstrate the effectiveness of our framework through a case study of an automated warehouse.
\end{abstract}



\section{Introduction}

\IEEEPARstart{C}{yber-physical systems} (CPSs) are widely used in critical areas such as transportation\cite{henshaw2018research}, healthcare\cite{zhang2015health}, and manufacturing\cite{wittenberg2016human}. However, as the number of networked components increases, CPSs become vulnerable to malicious cyber-attacks. High-profile incidents caused by cyber-attacks, such as the Stuxnet attack\cite{farwell2011stuxnet} and the Ukraine blackout\cite{choras2016cyber}, have resulted in significant economic losses and raised major concerns about the safety of networked CPSs.

To protect CPS from cyber-attacks, researchers have developed various detectors to detect intrusions\cite{Ding2018A}. However, these detectors are vulnerable to stealth attacks, which are designed with the knowledge of the systems to avoid detection\cite{musleh2019survey}. Using special stealth techniques, attackers can deceive the controller and bring systems into critical states\cite{cardenas2011attacks, khazraei2022resiliency, zhang2021stealthy}. For example, in the case of automatic vehicles, attackers can strategically hide attacks within sensor noise. Stealth attacks can maximize the state estimation error without triggering alarms, causing the vehicle to deviate from its desired trajectory and lead to collisions. It is evident that detectors have limited effectiveness when attacks are stealthy. However, due to the stealthiness of cyber-attacks, system behavior anomalies often emerge too late for safety controls to prevent incidents, as systems uncontrollably descend into critical states. This highlights the need for advanced risk assessment to provide early warning for safety controls. Therefore, beyond improving existing detection methods, we are fundamentally concerned with the question: How risk is a given system under stealth attacks?

\IEEEpubidadjcol

To quantify the risk of attacked systems, researchers study the responses of attacked systems to stealth attacks\cite{Yang2022A}. They use the performance indicators of the attacked systems as safety metrics, such as estimator error\cite{sui2020vulnerability}, recoverability\cite{liu2019joint}, reliability margin\cite{akametalu2014reachability}, sensitivity\cite{teixeira2019optimal}, resilience\cite{khazraei2022resiliency}, and reachability\cite{murguia2020security}. Among these metrics, reachability is particularly useful because it can formally describe the behavior of systems under attack. Murguia et al.\cite{murguia2020security} used the size of the attacker's reachable set and its minimum distance to the critical state set as a metric. Kwon et al.\cite{kwon2017reachability} assessed safety by comparing the reachable set and the safe region. Hwang et al.\cite{hwang2023lmi} quantified the potential risk of multi-agent systems under attack using reachability analysis. The reachability of attacked systems is analyzed using different set representations, including intervals\cite{fan2021improved}, ellipsoids\cite{zhang2020reachability, murguia2020security, hashemi2018comparison, kwon2017reachability}, and zonotopes\cite{liu2021reachability}. Note that these studies are built on the assumption of linear systems, and risk assessment for attacked nonlinear systems remains a challenge. The linear assumption limits their applicability in the real world, where most physical systems are inherently nonlinear\cite{wang2014iss}. Furthermore, scenarios may raise external risks to CPSs, which are not considered in the above studies. In a given scenario, most risk arises from the attacked system‘s interactions with other objects\cite{lyu2019safety}. For instance, an autonomous vehicle under attack can lead to more severe accidents if multiple pedestrians surround the vehicle. 

Motivated by the above-mentioned challenges, this article proposes a novel framework for assessing the risk to nonlinear systems under attack, considering relevant scenarios. As shown in Fig. \ref{paperflow}, the framework consists of a Stealth Reachability Analysis (SRA) algorithm, a construction method, and a Reachability and Risk field-based (RR) metric. We define the reachability of a system under stealth attacks as the Attacker's Stealth Reachable (ASR) set. The ASR set contains all the states that attackers can induce the system to reach. However, it is challenging to compute the exact bounds of the ASR set, especially for nonlinear systems\cite{kwon2017reachability}. To approximate the ASR set, we propose the SRA algorithm using multiple standard set representations. To formally describe the risk distribution of the scenario, we introduce a concept of risk field constructed by risk sets, where risk sets consist of the critical regions and loss expectations of possible events that the attacked system may trigger. Finally, we quantify the risk value by checking if the ASR set intersects with risk sets, which is the proposed RR metric. The idea of this metric is that if the ASR set intersects with a risk set, attackers can stealthily manipulate the system state into the critical region without being detected. The corresponding event will occur and cause harm. This way, the assessed risk value synthesizes the impact of system dynamics, noise, attacks, and scenarios on safety. In summary, the main contributions of this work are:
\begin{enumerate} 
\item{This work extends the reachability analysis of systems under stealth attacks to nonlinear systems. We propose an algorithm to approximate the ASR set. } \label{contribute1}
\item{We introduce a risk field concept to formally describe a scenario's risk distribution and propose a method for constructing the risk field.} \label{contribute2}
\item{We propose a dynamic metric based on the ASR set approximation and the risk field, which can dynamically quantify the risk for attacked systems within given scenarios.} 
\end{enumerate}

\begin{figure*}
  \centering
  \includegraphics[width=0.75\linewidth]{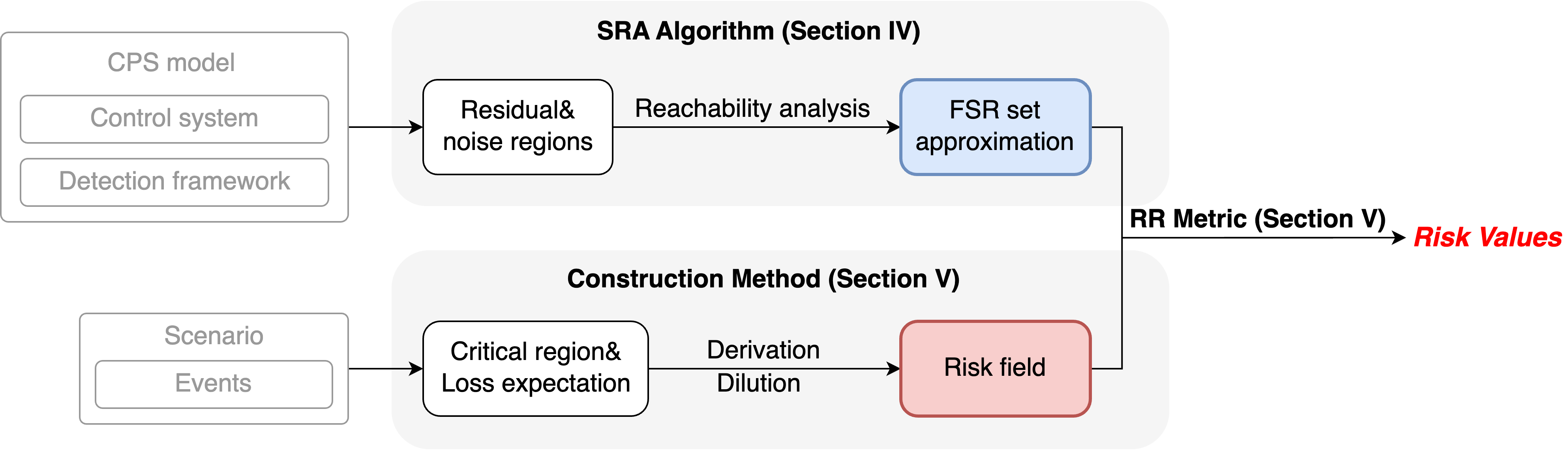}
  \caption{The framework to assess the risk of systems under stealth attacks.}
  \label{paperflow}
\end{figure*}

The structure of this article is as follows: In Section \ref{System description}, we describe the nonlinear system model under analysis, which includes aspects of the control system, state estimation, and detection. Section \ref{Analysis algorithm} introduces the algorithm used to approximate the ASR set, detailing the preliminaries, reachability analysis, and an overview of the algorithm. In Section \ref{Construction method and Metric}, we elaborate on the construction method for risk fields and present our proposed metric. Section \ref{Illustrative case} provides an illustrative case and a numerical simulation to demonstrate the effectiveness of our framework. Finally, we conclude in Section \ref{Conclusion}.

\textit{Notations}: In this article, $\mathbb{R}^{n}$ denotes the $n$-dimensional Euclidean space. An $n$-dimensional vector is represented as $a=[a^{[1]}\ a^{[2]}\ \dots\ a^{[n]} ]\in \mathbb{R}^n$. An $n\times m$ matrix is denoted by $A\in\mathbb{R}^{n\times m}$, with $A^{[i][j]}$ indicating the element at the $i$-th row and $j$-th column. $\mathbf A \subseteq \mathbb{R}^{n}$ represents a set or region in $n$-dimensional space. An $n$-dimensional interval is denoted by $[a, b]=\{x\in \mathbb{R}^{n}| a\leq x \leq b\}$, where $a,b\in \mathbb{R}^n$. The Minkowski sum of two sets is represented by $\oplus$, e.g., $\mathbf A \oplus \mathbf B =\{a+b \mid a\in \mathbf A, b\in\mathbf B \}$. The vector $x_k$ denotes a discrete point at time $t=t_k$, and $x_{0:n}=[x_0, \dots, x_n]$ refers to a sequence of $x_k$. A function $f: \mathbb{R}^{n_x} \to \mathbb{R}^{n_y}$ is Lipschitz continuous with a constant $L\ge 0$ if $\|  f(x_1) - f(x_2) \|\le L \| x_1 - x_2 \|$ for any $x_1, x_2\in \mathbb{R}^{n_x}$. Similar to the approximation between values, we also use the symbol $\approx$ to denote the over-approximation between sets, e.g., $\mathbf{A}\approx \mathbf{B}$, where $\mathbf{A}\subseteq \mathbf{B}$.

\section{Problem formulation}\label{System description}
This section introduces the control system and detection approach that our reachability analysis is based on. 
\subsection{Control system}
Consider a noisy nonlinear system described as follows:
\begin{align}
  x_{k+1} &= f(x_k, u_k) + w_k\label{dynamic}\\
  y_{k+1} &= h(x_k) + v_{k}\label{sensor}
\end{align}
where $x_k\in \mathbb{R}^{n_x}$ is the system state, $u_k\in \mathbb{R}^{n_u}$ is the control input during time interval $t\in[t_{k}, t_{k+1}]\subseteq \mathbb{R}$, $w_k\in \mathbb{R}^{n_x}$ is the process noise with a covariance $P_{w}\in \mathbb{R}^{n_x\times n_x}$, $y_{k+1} \in \mathbb{R}^{n_y}$ is the sensor output, and $v_{k}\in \mathbb{R}^{n_y}$ is the measurement noise with a covariance $P_{v}\in \mathbb{R}^{n_y\times n_y}$. The dynamic function $f: \mathbb{R}^{n_x} \times \mathbb{R}^{n_u}\to \mathbb{R}^{n_x}$ and the sampling function $h: \mathbb{R}^{n_x} \to \mathbb{R}^{n_y}$ are Lipschitz continuous and differentiable of high order.


In the presence of noise, the true system state $x_k$ is often inaccessible. Therefore, the controller uses the estimated state, $\hat{x}_k$, as its input. The controller is formulated as follows:
\begin{align}
  u_{k+1} &= g(\hat{x}_k)  \label{controller}
\end{align}
where $\hat{x}_k$ is the estimated state. Similarly, the function  $g:\mathbb{R}^{n_x}\to \mathbb{R}^{n_u}$ is Lipschitz continuous and differentiable of high order.

\subsection{state estimation}
We use the Unscented Kalman Filter (UKF) to estimate the system state as follows:
\begin{align}
  \hat{x}_{k} &= \hat{x}^-_{k}+ K_k(y_{k} - \hat y_k)\label{filter}
\end{align}
where $K_k\in \mathbb{R}^{n_x\times n_y}$ is the Kalman gain matrix, $\hat y_k\in \mathbb{R}^{n_y}$ is the predicted sensor output, and $\hat{x}^-_{k}\in \mathbb{R}^{n_x}$ is the prior predicted state. 

Note that non-standard distributions from nonlinear transformations are difficult to describe\cite{sarkka2007unscented, chambers1976method}. As a result, the state estimation relies on the Gaussian approximation assumption as follows:
\begin{assumption}[Gaussian approximation]
  Consider a general nonlinear function given by $y = f(x),\ x\sim \mathcal{N}(\mu_x, P_x) $ with $x \in \mathbb{R}^n$, $y \in \mathbb{R}^m$, and $f:\mathbb{R}^{n_x} \mapsto \mathbb{R}^{n_y}$. To approximate the non-standard distribution of $y$, we assume that $y$ follows a Gaussian distribution
  \begin{equation}
    \binom{x}{y}\sim \mathcal{N}\left(\binom{\mu_x}{\hat \mu_y},\begin{pmatrix}
      P_x &\hat P_{xy} \\
      \hat P_{xy}^T & \hat P_y
     \end{pmatrix} \right)  
  \end{equation}
\end{assumption}
where $\hat\mu_y=\mathrm{E}(y)$, $\hat P_y=\mathsf{Var}(y)$, and $\hat P_{xy}=\mathsf{Cov}(x,y)$ are the estimated mean, variance, and covariance.

Based on the Gaussian approximation, the parameters in \eqref{filter} are iterated by
\begin{align}
    [\hat{x}^-_{k}, \hat P_{x_{k}}^-, \hat P_{y_{k}}, \hat P_{xy_{k}}] &= \mathsf{UT}_{u_k,P_{w},P_{v}}(\hat{x}_{k-1}, \hat P_{x_{k-1}})\label{UT}\\
  K_k &=\hat P_{xy_{k}}(\hat P_{y_{k}})^{-1}\label{UKF-1}\\
  \hat P_{x_{k}} &= \hat P_{x_{k}}^- - K_{k} \hat P_{y_{k}} K_{k}^T\label{UKF-2}
\end{align}
where $\hat P_{x_{k}}^-$ is the prior predicted variance of $x_k$, $\hat P_{x_{k}}$ is the estimated variance of $x_k$, $\hat P_{y_{k}}$ is the predicted variance of $y_k$, $\hat P_{xy_{k}}$ is the predicted covariance of $x_k$ and $y_k$, and $\mathsf{UT}$ denotes the Unscented Transform (UT) algorithm. Generate the sigma points set $\sigma_x$ as:
\begin{equation}
        \begin{split}
      \sigma^{[0]} &\!=\! \hat{x}_{k-1}\\
      \sigma^{[i]} &\!=\! \hat{x}_{k-1}\!+\!\sqrt{(n_x\!+\!\lambda)} \sqrt{\hat P_{x_{k-1}}}^{[i]}\\
      \sigma^{[n_x+i]} &\!=\! \hat{x}_{k-1}\!-\!\sqrt{(n_x\!+\!\lambda)} \sqrt{\hat P_{x_{k-1}}}^{[i]}
     \end{split} \ \  i \!=\! 1, \dots, n_x
\end{equation}
where $\lambda$ serves as a scaling parameter. The notation $\sqrt{P}$ represents the Cholesky factorization \cite{golub2013matrix} of the covariance matrix $P$. The process of  $\mathsf{UT}$ can be given as follows:
\begin{align}
      \hat{x}^-_{k}&\!=\!\sum^{2n_x}_{i=0} \omega^{[i]}f(\sigma_x^{[i]},u_k)\label{UT2-1}\\
      \!\hat P_{x_{k}}^- &\!=\! \sum^{2n_x}_{i=0} \! \nu^{[i]}\!\left(f(\!\sigma_x^{[i]}\!,\!u_k\!)\!-\!\hat{x}^-_{k}\right) \! \left(f(\!\sigma_x^{[i]}\!,\!u_k\!)\!-\!\hat{x}^-_{k}\right)^T\!+\!P_{w}\label{UT2-2}\\
      \hat P_{y_{k}} &\!=\! \sum^{2n_x}_{i=0} \nu^{[i]}\left(h(\sigma_x^{[i]})-\hat{x}^-_{k}\right)\!\left(h(\sigma_x^{[i]})-\hat{x}^-_{k}\right)^T\!+\!P_{v}\!\label{UT2-3}\\
      \hat P_{xy_{k}}&\!=\! \sum^{2n_x}_{i=0} \nu^{[i]}\!\left(\!\sigma_x^{[i]}\!-\!\hat{x}_{k-1}\!\right)\!\left(\!h(\sigma_x^{[i]})\!-\!\hat{x}^-_{k}\!+\!\sqrt{ P_v}^{[i]}\!\right)^T\label{UT2-4}
\end{align}
where $\omega$ and $\nu$ are the associated weights (see \cite{sarkka2007unscented}). For simplicity, we denote the above equations \eqref{UT2-1}, \eqref{UT2-2}, \eqref{UT2-3}, and \eqref{UT2-4} by $ f_{\hat{x}_{k}^-}$, $f_{\hat{P}_{x_{k}}^{-}}$, $h_{\hat{P}_{y_{k}}}$, and $h_{\hat{P}_{xy_{k}}}$ respectively. More detailed information about the UKF and UT algorithm can be found in\cite{sarkka2007unscented}.

\subsection{Detection}
As a representative example, we utilize a chi-square detector to detect attacks in our system model. Note that our analysis can be extended to other detectors with suitable adjustments. We introduce a chi-square detector based on the residual, which is a simplified version from \cite{wu2023double}. 

Define the residual 
\begin{align}\label{residual}
    r_k:=y_{k} - \hat y_k
\end{align}
with the estimated sensor output $\hat{y}_{k}=h(\hat{x}_{k-1})$. Julier et al. \cite{julier2004unscented} and Zhang at al. \cite{zhang2020reachability} have demonstrated that the Kalman filter provides a minimum variance unbiased estimation for system states. Under the assumption of Gaussian approximation in nonlinear systems, $\hat{x}_{k-1}$ serves as the unbiased estimate of $x_{k-1}$. When the filter stabilizes in the absence of attacks, the mean and variance of $r_k$ are $\mu_{r_k}=\mathrm{E}({y}_k - \hat y_k)\approx\mathbf{0}$ and ${P}_{r_k}=\mathrm{Var}({y}_k - \hat y_k)\approx\hat{P}_{y_k}$, respectively. Therefore, if the UKF stabilizes without attacks, the following conclusion is approximately satisfied $r_k\sim \mathcal{N}(\mathbf{0},\hat{P}_{r_k})$. Based on this conclusion, we can detect attacks by testing the following two incompatible statistical hypotheses:
\begin{equation}
  \label{hypotheses}
  \begin{cases}
    \mathcal{H}_0: {r}_k \sim \mathcal{N}(\mathbf{0}, \hat P_{r_k}),\ \text{Attack-free},\\
    \mathcal{H}_1: {r}_k \nsim \mathcal{N}(\mathbf{0}, \hat P_{r_k}),\ \text{Attacked}.
  \end{cases}
\end{equation}
Based on this statistical hypothesis, we define a chi-square detector
\begin{equation}
  \label{alert}
  Alarm_k=\left\{\begin{matrix}
    0 \ \mathrm{and}\ \mathrm{accept}\ \mathcal{H}_0,&\mathrm{if}\ {r}_{k}^T \hat{P}_{r_k}^{-1} {r}_{k}\le \varepsilon_r\\ 
    1 \ \mathrm{and}\ \mathrm{accept}\ \mathcal{H}_1,&\mathrm{if}\ {r}_{k}^T \hat{P}_{r_k}^{-1} {r}_{k}> \varepsilon_r  \end{matrix}\right.
\end{equation}
with a given threshold $\varepsilon_r\in\mathbb{R}$. We assume that $\varepsilon_r=\chi^2_{n_y}(1-\alpha)$ , where $\alpha\in\mathbb{R}$ is the false alarm probability and $\chi^2_{n}$ is the cumulative distribution function of a $n$-dimensional chi-square distribution. 
The tolerance region $\mathbf{\Gamma}_k \subset \mathbb{R}^{n_y}$ of this detector can be represented as follows:
\begin{equation}
    \begin{split}
    \mathbf{\Gamma}_k&:= \{{r}_{k}\in \mathbb{R}^{n_y}| {r}_{k}^T \hat{P}_{r_k}^{-1} {r}_{k}\le \varepsilon_r \}
    \end{split}.
\end{equation}
The probability of this detector making a type I error is
  \begin{equation}
    \mathrm{Pr} ({r}_{k} \notin \mathbf{\Gamma}_k|\text{Attack-free})= \alpha.
  \end{equation}
The threshold value $\varepsilon_r$ determines the detector's sensitivity, which should be selected with a tolerable false alarm probability.


\subsection{The system under stealth attacks}\label{Stealth attacks}
Stealth attacks pose a major threat to networked CPS. As shown in Fig. \ref{CPS_attacks}, Attackers can disrupt state estimation by injecting false data into sensors. The increasing error between estimated and true states may cause incorrect control commands, pushing the system into unsafe conditions and raising the risk of accidents. We base our analysis on an assumption about attacker capabilities as follows:

\begin{figure}
  \centering
  \includegraphics[width=1.0\linewidth]{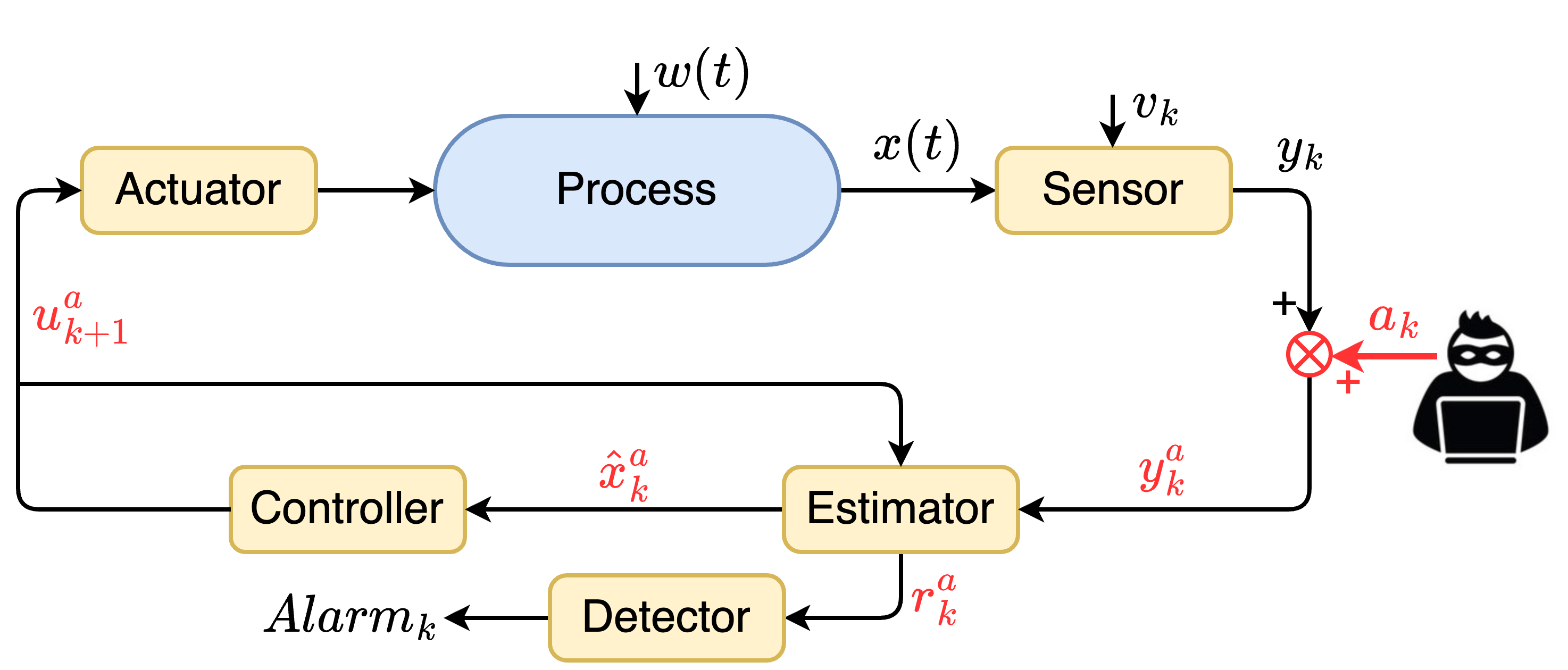}
  \caption{An attacked system architecture: by injecting false data into the measurements, an attacker can intrude into the system without being detected.}
  \label{CPS_attacks}
\end{figure}

\begin{assumption}[Attackers' Ability]\label{Attackers' Ability}
The attacker has complete knowledge of the system model and detection framework, and they can obtain real-time values of all system variables instantly. Additionally, attackers have almost unlimited computing resources, which they can use to design attack sequences that best serve their goals.
\end{assumption}

This assumption implies that the attackers can maximize the effectiveness of their attacks while remaining undetected\cite{kwon2017reachability}. We assume that the system is attacked during $t\in [t_1,t_n]$, and the state estimation has operated for a sufficient time (i.e., the UKF is already steady before $t=t_0$).
The system under additive sensor attacks can be characterized as follows:
\begin{equation}\label{attacked system}
\begin{cases}
  x_{k+1}^a = f(x_k^a, u_k^a) + w_k\\
  y_{k+1}^a = h(x^a_k) + v_{k}+ a_k\\
  \hat{x}^a_{k} = \hat{x}^{a-}_{k}+ K^a_kr^a_k\\
  u^a_{k+1} = g(\hat{x}^a_k)
\end{cases}
\ \ k=1,\dots, n
\end{equation}
where $a_{k}\in\mathbb{R}^{n_y}$ symbolizes the attack vector, which is injected into the sensors at time $t=t_k$. The superscript $^a$ is used to denote variables affected by the attacks.

We can rewrite the residual of the attacked system as $r^a_{k-1}=h(x^a_k)-h(\hat{x}^a_k)+v_k+a_k$. The attack influences the residual distribution, implying that the detector constrains the attacker’s impact on the system. Based on the detector described in \eqref{alert}, the set of stealth attacks is given as:
\begin{align}
    \mathbf{\mathcal{A}}_{0:n}=\left \{ a_{1:n}\in \mathbb{R}^{n_y\times n} \left|  r_k^a\in \mathbf{\Gamma}_k,k=1,\dots n \right .\right \}. 
\end{align}
This equation implies that our analysis does not need to consider the specific strategy used by attackers. Instead, we focus our analysis on the region of $r^a_k$.

Stealth attacks pose a more significant threat than detected attacks because they are challenging to detect. We analyze the system's reachability under stealth attacks in the following, which is the first step of our risk assessment.

\section{Reachability Analysis}\label{Analysis algorithm}
This section discusses the set theory, ASR set over-approximation, and the overview of our SRA algorithm.

\subsection{Preliminaries}\label{Set over-approximation}
To simplify the complex reachability analysis, we employ standard set representations in our SRA algorithm (as shown in Fig. \ref{approximate_algorithm}) to over-approximate the non-standard reachable sets.

\subsubsection{Set representation}
We use Zonotope for visualization and intersection detection due to its computational efficiency. A zonotope is defined as the image of a hypercube following an affine transformation. It is typically characterized by the Minkowski sum of its line segment generators.

\begin{definition}[Zonotope]
Denote the $n$-dimensional unitary interval $\mathbf{B}^m=[-1, 1]^m$. A zonotope $\mathcal{Z}\subset \mathbb{R}^{n}$ with $m$ generators $g_i\in \mathbb{R}^n$ ($i=1,2,\dots,m$) is defined as
\begin{equation}\label{zonotope_definition}
    \begin{split}
  \mathcal{\mathcal{Z}}&=<c, G>_Z\\
  &=\left\{c+Gx\in \mathbb{R}^n\left|x\in \mathbf{B}^m\right.\right\}
    \end{split}
\end{equation}
where $c\in \mathbb{R}^n$ is the center of  $\mathcal{Z}$ and $G=[g_1, g_2, \dots, g_m]\in \mathbb{R}^{n\times m}$ is the generator matrix. 
\end{definition}
The Minkowski sum and linear map of zonotopes have the following properties:
\begin{align}
  \!<\!c_1, H_1\!>_Z\!\oplus\! <\!c_2, H_2\!>_Z\! &= <\!c_1\!+\!c_2, [H_1\ H_2]\!>_Z\\
  L<c, H>_Z + b&= <Lc+b, LH>_Z\label{Zonotope_affine}
\end{align}
where $L\in \mathbb{R}^{l\times n}$ and $b\in \mathbb{R}^l$ are any map matrix and vector respectively.

Note that a zonotope is an interval when all of its generators are axial vectors. An interval also can be considered as a zonotope, e.g., $[a, b]=<{(a+b)}/{2}, \mathsf{diag}([{(b-a)}/{2}])>_Z$. 

Taylor model has been observed to scale quite well and produce accurate results for nonlinear systems\cite{chen2012taylor}. Hence, we use the Taylor model to approximate the reachability of the nonlinear dynamic. We use the definition of the Taylor model from \cite{chen2012taylor} as follows:

\begin{definition}[Taylor model]\label{Taylor model}
A Taylor model $\mathcal{TM}\subset \mathbb{R}^n$ of order $\kappa > 0$ is defined as follows:
\begin{equation}
    \begin{split}
            \mathcal{TM}&=<p,\mathbf{I}>_{TM}\\
    &=\left\{p(x)\in \mathbb{R}^{n_x}| x\in \mathbf{X}\right\}\oplus \mathbf{I}
    \end{split}
\end{equation}
where $p:\mathbb{R}^{m}\to \mathbb{R}^{n}$ is a polynomial function, $\mathbf{X}\subset \mathbb{R}^{n_x}$ is the domain, and $\mathbf{I}\subseteq \mathbb{R}^{n}$ is an interval used for remainder estimation. 
\end{definition}

Consider a function $y=f(x)$, where $f: \mathbb{R}^{n_x} \to \mathbb{R}^{n_y}$ is $\kappa$-order differentiable over a domain $\mathbf{X}$. Its $\kappa$-order Taylor expansion at a point $c\in \mathbf{X}$ can be given by
\begin{align}\label{Taylor expansion}
\begin{split}
        &f(x)=\underbrace{\sum_{i=0}^\kappa \sum_{j_1+\dots j_{n_x}=i}\left(\dfrac{f^{(i)}(c)}{i!}\prod_{l=1}^{n_x}(x^{[l]}-c^{[l]})^{j_l}\right)}_{p_f(x)} + \\
    &\underbrace{\sum_{j_1+\dots j_{n_x}=\kappa+1}\left(\dfrac{f^{(\kappa+1)}(\xi(x))}{(\kappa+1)!}_{x=c}\prod_{l=1}^{n_x}(x^{[l]}-c^{[l]})^{j_l}\right)}_{r_f(x)}
\end{split}
\end{align}
where $\xi(x)\in (c, x)$, the superscript $^{[l]}$  represents the $l$-th component of the vector, $p_f$ is the Taylor polynomial of $f$, $r_f$ is the remainder of $f$, and $f^{(i)}(c)$ represents the $i$-order derivative of $f$ at point $x=c$. $f^{(i)}(c)$ can be computed by
\begin{equation}
    \begin{split}
    f^{(i)}(c)=\left.\dfrac{\partial^i f(x)}{\partial [x^{[1]}]^{j_1}\dots\partial [x^{[n_x]}]^{j_{n_x}}}\right|_{x=c}.
\end{split}
\end{equation}
If there exists an interval $\mathbf{I}_f$ that satisfies $r_f(x)=f(x)-p_f(x)\in  \mathbf{I}_f$ for all $x \in \mathbf{X}$, then $<p_f, \mathbf{I}_f>_{TM}$ is an over-approximation of $f$ on $\mathbf{X}$, denoted as $f(\mathbf{X}) \subseteq p_f(\mathbf{X}) \oplus \mathbf{I}_f$. The interval $\mathbf{I}_f$ severs as the over-approximation of the remainder $r_f(x)$ over $\mathbf{X}$. It can be computed using interval arithmetic interval arithmetic (see \cite{moore2009introduction} for details), denoted by $\mathbf{I}_f = \mathsf{Int}(p_f)$. Next, we discuss the binary operations of the Taylor model from \cite{chen2015reachability}.

\begin{lemma}[Taylor model arithmetic]\label{Taylor model arithmetic}
    For any two Taylor models $<p_1, \mathbf{I}_1>_{TM}$ and $<p_2,\mathbf{I}_2>_{TM}$ over the same domain $\mathbf{X}$, the operation of sum and product between them can be easily computed by
\begin{align}
<p_1, \mathbf{I}_1>_{TM}+&<p_2, \mathbf{I}_2>_{TM}\nonumber\\
&= <p_1+p_2, \mathbf{I}_1\oplus\mathbf{I}_2>_{TM}\\
<p_1, \mathbf{I}_1>_{TM}\cdot &<p_2, \mathbf{I}_2>_{TM}\nonumber\\
&= <p_1\cdot p_2, \mathbf{I}_1\odot\mathsf{Int}(p_2)\oplus\mathbf{I}_2\odot\nonumber\\
        &\mathsf{Int}(p_1)\oplus\mathbf{I}_1\odot\mathbf{I}_2)>_{TM}\\
                <p_1, \mathbf{I}_1>_{TM} / &<p_2, \mathbf{I}_2>_{TM}\nonumber\\
    &=<p_1, \mathbf{I}_1>_{TM} \cdot <p_3, \mathbf{I}_3>_{TM}
\end{align}
where $<p_3, \mathbf{I}_3>_{TM}=p_{\frac{1}{x}}(p_2(\mathbf{X})\oplus\mathbf{I}_2)\oplus\mathsf{Int}(r_{\frac{1}{x}}(p_2(\mathbf{X})\oplus\mathbf{I}_2))$. 
\end{lemma}
\begin{proof}
    See \cite{chen2015reachability}.
\end{proof}

The two Taylor models share the same domain in the binary operations above. For Taylor models with different domains, operations are represented by $\oplus$, $\odot$, and $\oslash$. For example, $<p_1, \mathbf{I}_1>_{TM}\oplus<p_2, \mathbf{I}_2>_{TM}=p_1(\mathbf{X}_1)\oplus p_2(\mathbf{X}_2)\oplus\mathbf{I}_1\oplus\mathbf{I}_2$.

\subsubsection{Over-approximation}
This subsection introduces the conversion and over-approximation methods for these set representations.

\begin{proposition}[Zonotope to TM]
In a $n$-dimensional space $\mathbf{R}^n$, a zonotope can be converted to a $1$-order Taylor model when it has no more than $n$ non-axial generators.
\end{proposition}
\begin{proof}
    Assume that a $n$-dimensional zonotope  $\mathcal{Z}=<c,[g_1,\dots, g_n, g_{n+1},\dots, g_m]>_{\mathcal{Z}}$, where $g_1,\dots,g_n\in\mathbb{R}^n$ are non-axial vectors, and $g_{n+1},\dots,g_m\in\mathbb{R}^n$ are axial vectors. According to the definition of zonotope in \eqref{zonotope_definition}, $\mathcal{Z}$ can be rewritten with Minkowski sum as
\begin{equation}
    \begin{split}
\mathcal{Z}
    &=\Bigg\{\underbrace{ c+\sum_{i=1}^n x^{[i]} g_i}_{c+G_1x}\in \mathbb{R}^n|x^{[i]}\in [-1,1]\big\}\\
    &\oplus\Bigg\{\underbrace{\sum_{i=n+1}^m x^{[i]} g_i}_{G_2x}\in \mathbb{R}^n|x^{[i]}\in [-1,1]\Bigg\}\\
    &=p(\mathbf{B}^n)\oplus\mathbf{I}
    \end{split}
\end{equation}
where $G_1=[g_1,\dots, g_n]$ is the $1$-order coefficient matrix of $p$, $G_2=[g_{n+1},\dots, g_m]$ is the parameters matrix of $\mathbf{I}$.
\end{proof}
For a zonotope with more than $n$ generators, we use the method from \cite{kopetzki2017methods} to reduce the number of its generators to $n$. Then, this zonotope can be converted to a Taylor model.

For a Taylor model of order $\kappa\ge 2$, we over-approximate it with a zonotope using the method from \cite{chen2012taylor}. The linear part of the Taylor model is converted to a zonotope. The remaining non-linear part (the sum of terms in the polynomial with order $2$ or above) is over-approximated as an interval using the interval algorithm. Finally, the Minkowski sum of the zonotope and interval is used as an over-approximation of the Taylor model, which is also a zonotope. We use $\mathcal{Z}=\mathsf{TM2Z}(\mathcal{TM})$ to represent this process.

Based on the above conclusion, we propose a method to over-approximate an ellipsoid with a Taylor model.

\begin{proposition}[Ellipsoid to Taylor model]\label{E2TM}
    A Taylor model can over-approximate an n-dimensional ellipsoid with complexity $\mathcal{O}(n^3)$.
\end{proposition}
\begin{proof}
Consider an ellipsoid $\!\mathbf{\mathcal{E}}\!=\!\left\{\!x\!\in\!\mathbb{R}^{n}\! \mid\! (x\!-\!c)^T Q (x\!-\!c)\!\le\! 1\!\right\}$. The shape matrix $Q\!=\!P\Lambda P^T$ is symmetric, where $P$ represents the matrix of eigenvectors and $\Lambda$ is the diagonal matrix of eigenvalues of $Q$. The ellipsoid $\mathbf{\mathcal{E}}$ can be over-approximated by a hyper-rectangle $\mathcal{H}_1=\{x\in\mathbb{R}^{n} \mid c+H_1x\}$ aligned with the axes of $\mathcal{E}$, where $H_1=(\Lambda^{-0.5}P)^T$ and $\Lambda^{-0.5}$ denotes the inverse square root of $\Lambda$. Concurrently, the ellipsoid can also be over-approximated by another hyper-rectangle $\mathcal{H}_2=\{x\in\mathbb{R}^{n} \mid c+H_2x\}$, with $H_2$ consisting of the tangency vectors $[h_1,\dots,h_n]$, each aligned with the $i$-th axis of the coordinate system. The intersection of $\mathcal{H}_1$ and $\mathcal{H}_2$ forms a zonotope $\mathcal{Z}$, centered at $c$, denoted as $\mathcal{Z}=\mathcal{H}_1\cap\mathcal{H}_2=\langle c,G\rangle_{\mathcal{Z}}$, where $G=[g_1,\dots, g_n, g_{n+1},\dots, g_{2n}]$. Here, $[g_1,\dots,g_n]$ are aligned with the axes of the ellipsoid, and $[g_{n+1},\dots,g_{2n}]$ with the coordinate system. Given that $\mathcal{E}\subseteq \mathcal{H}_1$ and $\mathcal{E}\subseteq \mathcal{H}_2$, by subset transitivity, we have $\mathcal{E}\subseteq \mathcal{Z}$. Ultimately, $\mathcal{Z}$ can be transformed into a first-order Taylor model $\mathcal{TM}$, as it has no more than $n$ non-axial generators.

The computation of $H_1$, $H_2$, and $G$ requires complexities of $\mathcal{O}(n^3)$, $\mathcal{O}(n^2)$, and $\mathcal{O}(n)$, respectively. Therefore, the overall computational complexity of this process is $\mathcal{O}(n^3)$. For brevity, we denote this process as $\mathcal{TM}=\mathsf{E2TM}(\mathcal{E})$.
\end{proof}

\subsection{Reachability analysis}\label{Estimated State Set}


In control theory, reachability is defined as the system’s ability to transition from one state to another within a finite time. However, the direct computation of its precise boundaries often presents a challenge. To analyze the reachability of the system under stealth attacks, we theoretically outline the definition of the ASR set.

\begin{definition}[ASR set] \label{ASR set}
    Assuming that the time interval that the system is under attack is $\mathbf{T}^a_{1:n}:=[t_1, t_n]$ and the initial state is $x_0\approx \hat{x}_0$, then the ASR set can be defined by
\begin{equation}\label{continuous_ASR_set}
\begin{split}
  &\mathcal R^a_{f}(\mathbf{X}_0, \mathbf{T}^a_{1:n}) :=\\
    &\left\{x^a\in \mathbb{R}^{n_x} \left| \begin{split}
      &\mathrm{Eq.}\ \eqref{attacked system},\\
      &x^a_0=x_0, w_k\in \mathbf{W}_k ,{r}^a_{k}\in \mathbf{\Gamma}_k,\\
      &\text{and}\ k=1,2,\dots, n.
      \end{split} \right.\right\}
\end{split}
\end{equation}
where $\mathbf{W}_k$ is the region of  process noise. 
\end{definition}

Note that systems with random variables have infinite reachability because unbounded randomness allows the system to be in any state (although with a small probability). To solve this problem, researchers usually define reachable sets based on statistical properties, i.e., limit the value range of random variables with a threshold probability \cite{kwon2017reachability, zhang2020reachability}. In this article, we recommend that $\mathbf{W}_k:=\{w_k\in \mathbb{R}^{n_x}|w_k^TP_{w}w_k\le \varepsilon_w\}
$ and $\varepsilon_w=\chi^2_{n_x}(\beta)$. In this case, we analyze the reachability under the noise conditions with the probability $\beta$.

The definition of \eqref{continuous_ASR_set} is theoretical, and its precise bounds are difficult to determine. To simplify the approximation of the ASR set, we redefine it with discrete time steps. 
\begin{theorem}\label{state_set}
    Assume that $\mathbf{U}^a_k$ is the set of control inputs at $t=t_k$. The ASR set can be over-approximated by
\begin{align}
     \mathcal R^a_{f}(\mathbf{X}_0, &\mathbf{T}^a_{1:n})\approx \bigcup_{k=1}^{n}  p_f(\mathbf{X}^a_{k}, \mathbf{U}^a_k)\oplus \mathbf{I}_f \oplus \mathbf{W}_k
\end{align}
where $p_f$ is the Taylor polynomial of $f$ in \eqref{Taylor expansion}.
\end{theorem}
\begin{proof}
The reachable set of the system state at $t=t_{k+1}$ can be given as
\begin{equation}
\begin{split}
    \mathcal R^a_{f}&(\mathbf{X}_k^a, t_{k+1}) := \\
    &\left\{{x}^a_{k+1}\in \mathbb{R}^{n_x} \left| 
    \begin{split}
      &x^a_{k+1}=f(x^a_{k}, u^a_k) + w_k \\
      &x^a_k\in \mathbf{X}_k^a, u^a_k \in \mathbf{U}^a_k, w_k\in \mathbf{W}_k.
    \end{split} \right.\right\}.
\end{split}
\end{equation}
There exists an interval $\mathbf{I}_f$ such that: 
\begin{align}
    &f(x^a_k, u^a_k)+w_k\in p_f(x^a_k, u^a_k) + w_k \oplus \mathbf{I}_f\\
    &\forall x^a_k\in \mathbf{X}^a_k, \ u^a_k\in \mathbf{U}^a_k, and\ w_k\in \mathbf{W}_k.\nonumber
\end{align}
According to the definition of Taylor model approximation, there exists
\begin{align}
    \mathcal R^a_{f}(\mathbf{X}_k^a, t_{k+1})&=f(\mathbf{X}^a_{k}, \mathbf{U}^a_k)\oplus \mathbf{I}_f \oplus \mathbf{W}_k\nonumber \\
    &\subseteq p_f(\mathbf{X}^a_{k}, \mathbf{U}^a_k)\oplus \mathbf{I}_f \oplus \mathbf{W}_k.
\end{align}
Considering $\mathcal R^a_{f}(\mathbf{X}_0, \mathbf{T}^a_{1:n})= \bigcup_{k=1}^{n-1} \mathcal R^a_{f}(\mathbf{X}_k^a, t_{k+1})$, we can deduce this proposition.
\end{proof}
For simplicity, we denote $\approx$ represents the over-approximation process, e.g., $\mathbf{X}_{k+1}^a\approx p_f(\mathbf{X}^a_{k}, \mathbf{U}^a_k)\oplus \mathbf{I}_f \oplus \mathbf{W}_k$. 

The set of control inputs $\mathbf{U}^a_k$ is defined as
\begin{equation}\label{controller_set}
\begin{split}
    \mathbf{U}^a_{k+1} =
    \left\{u^a_{k+1}\in \mathbb{R}^{n_u} \left| 
    \begin{split}
      &u^a_{k+1} = g(\hat{x}^a_{k}), \hat{x}^a_{k}\in \hat{\mathbf{X}}^a_{k}.
    \end{split} \right.\right\}
\end{split}
\end{equation}
where $\hat{\mathbf{X}}^a_{k}$ is the set of  estimated states. We can approximate $\mathbf{U}^a_{k+1}$ defined in \eqref{controller_set} as
\begin{align}
    \mathbf{U}^a_{k+1}\approx p_g(\hat{\mathbf{X}}_k^a)\oplus\mathbf{I}_g.
\end{align}
$\hat{\mathbf{X}}^a_{k}$ can be defined as the following:
\begin{equation}\label{estimated states set}
\begin{split}
        \hat{\mathbf{X}}&^a_{k}=\left\{ \hat{x}^a_{k}\in \mathbb{R}^{n_x}\left|
    \begin{split}
       &\hat{x}^a_{k}=\hat{x}^{a-}_{k} + K^a_k r^a_k ,\\
       &\mathrm{Eq}.\ \eqref{UT}, \eqref{UKF-1}, \eqref{UKF-2}, \\
        &\hat{x}^{a}_{k-1}\in \hat{\mathbf{X}}^{a}_{k-1}, \hat{P}_{x_{k-1}}^a\in \hat{\mathbf{P}}_{x_{k-1}}^a,\\
        &\text{and}\ r^a_k \in \mathbf{\Gamma}_k.
    \end{split}
    \right.\right\}
\end{split}
\end{equation}
where $\hat{\mathbf{X}}^{a-}_{k}\subset \mathbb{R}^{n_x}$ is the reachable set of predicted state from the UT and $\mathbf{K}^a_k\subset \mathbb{R}^{n_x\times n_y}$ is the reachable set of Kalman gain matrix from the UKF. To approximate $\hat{\mathbf{X}}_{k+1}^a$, there is a critical issue that needs to be resolved: the reachability of UKF. In the following, we will introduce the approach to analyze the reachability of the UT.

\begin{proposition}
\label{estimated states set Over-approximation}
The estimated state set defined in \eqref{estimated states set} can be over-approximated as
\begin{align}
    \hat{\mathbf{X}}_{k}^a&\approx p_{\hat{x}_{k}^a}(\hat{\mathbf{X}}^a_{k-1},\hat{\mathbf{P}}_{x_{k-1}}^a, \mathbf{U}^a_{k-1}, \mathbf{B}^{n_y})\oplus \mathbf{I}_{\hat{x}_{k}^a}
\end{align}
with
\begin{align}
    p_{\hat{x}_{k}^a}&=p_{f_{\hat{x}_{k}^-}}+p_{h_{\hat{P}_{xy_{k}}}}\cdot p_{A^{-1}}(p_{h_{\hat{P}_{y_{k}}}})\odot p_{\mathbf{\Gamma}_k}\\
    \mathbf{I}_{\hat{x}_{k}^a}&=\mathbf{I}_{f_{\hat{x}_{k}^-}}\oplus \mathbf{I}_{h_{\hat{P}_{xy_{k}}}} \odot \mathsf{Int}(p_{A^{-1}}(p_{h_{\hat{P}_{y_{k}}}}))\odot \mathsf{Int}(p_{\mathbf{\Gamma}_k})\nonumber\\
    &\oplus \mathbf{I}_{A^{-1}(p_{h_{\hat{P}_{y_{k}}}})}\odot\mathsf{Int}(p_{h_{\hat{P}_{xy_{k}}}})\odot \mathsf{Int}(p_{\mathbf{\Gamma}_k})\nonumber\\
    &\oplus\mathbf{I}_{\mathbf{\Gamma_k}}\odot \mathsf{Int}(p_{h_{\hat{P}_{y_{k}}}}) \odot\mathsf{Int}(p_{h_{\hat{P}_{xy_{k}}}})\nonumber\\
    &\oplus \mathbf{I}_{h_{\hat{P}_{xy_{k}}}} \odot \mathbf{I}_{A^{-1}(p_{h_{\hat{P}_{y_{k}}}})}\odot\mathbf{I}_{\mathbf{\Gamma}_k}
\end{align}
where the subscript $p_{A^{-1}}$ denotes the polynomial function associated with the matrix inverse. $p_{f_{\hat{x}_{k}^-}(\cdot)}$, $p_{h_{\hat{P}_{y_{k}}}(\cdot)}$, and $p_{h_{\hat{P}_{xy_{k}}}(\cdot)}$ is the polynomial functions of  \eqref{UT2-1}, \eqref{UT2-3}, and \eqref{UT2-4}, respectively, defined over the domains $\hat{\mathbf{X}}^a_{k-1}$, $\hat{\mathbf{P}}_{x_{k-1}}^a$, and $\mathbf{U}^a_{k-1}$. It is also noted that the domain of $p_{\mathbf{\Gamma}_k}$ is $\mathbf{B}^{n_y}$.
\end{proposition}
\begin{proof}
    Reviewing the definition of the residual in \eqref{residual} and the attacked system in \eqref{attacked system}, the necessary and sufficient condition for attacks to remain stealthy is $r^a_{k+1}\in \mathbf{\Gamma}_{k}$. We can over-approximate $\mathbf{\Gamma}_{k}$ with a Taylor model as $\mathbf{\Gamma}_{k+1}\approx p_{\mathbf{\Gamma}_{k+1}}(\mathbf{B}^{n_y})\oplus\mathbf{I}_{\mathbf{\Gamma}_{k}}$ (see Proposition \ref{E2TM}). Then \eqref{estimated states set} can be rewritten as $\hat{\mathbf{X}}_{k+1}^a=\hat{\mathbf{X}}_{k+1}^{a-}+\mathbf{K}^a_{k+1}\odot \mathbf{\Gamma}_{k+1}$, where $\mathbf{K}^a_{k+1}=\hat{\mathbf{P}}^a_{xy_{k+1}}[\hat{\mathbf{P}}^a_{y_{k+1}}]^{-1}$. Applying the corresponding Taylor model over-approximation to it and using Lemma \ref{Taylor model arithmetic}, we can obtain the conclusion of Proposition \ref{estimated states set Over-approximation}. Additionally, $\hat{\mathbf{P}}^a_k$ can be updated in the same way, as $\hat{\mathbf{P}}^a_k=\hat{\mathbf{P}}_{x_{k+1}}^{a-} - \mathbf{K}^a_{k+1}\odot \hat{\mathbf{P}}_{x_{y+1}}^{a}\odot [\mathbf{K}^a_{k+1}]^T$.
\end{proof}

Proposition \ref{estimated states set Over-approximation} employs a lot of Taylor expansion processes, which leads to unnecessary computational consumption. However, as the filter stabilizes, the estimated state variance tends to be constant. We can simplify the computation using the following assumption.

\begin{assumption}\label{UKF_center}
    For the UKF, there are $\hat{P}_{x_k}^{a-}\approx \hat{P}_{c_k}^{a-}$ for all $\hat x^a_k\in \hat{\mathbf{X}}_k^a$ , where $c_k$ is the center point of $\hat{\mathbf{X}}_k^a$.
\end{assumption}

Based on this assumption, the result is not significantly different from that of Proposition \ref{estimated states set Over-approximation}. In this way, we can simplify the approximation of $\hat{\mathbf{X}}_k^a$ as follows: 
\begin{theorem}
    Under Assumption \ref{UKF_center}, the over-approximation in Proposition \ref{estimated states set Over-approximation} can be simplified to
\begin{align}\label{simplified estimated state set}
    \hat{\mathbf{X}}_{k+1}^a&\approx \hat{\mathbf{X}}_{k+1}^{a-} \oplus K^a_k \mathbf{\Gamma}_k\nonumber\\
    &=p_{f_{\hat{x}_{k+1}^-}}(\hat{\mathbf{X}}^a_{k},\hat{P}_{x_k}^a)\oplus K^a_k p_{\mathbf{\Gamma}_k}(\mathbf{B}^{n_y})\nonumber\\
    &\oplus\mathbf{I}_{f_{\hat{x}_{k+1}^-}}\oplus\mathbf{I}_{\mathbf{\Gamma}_k}
\end{align}
where $\hat{P}_{x_k}^a$ and $K^a_k$ are computed by UKF with $\hat{x}^a_k=\mathsf{center}(\hat{\mathbf{X}}_k^a)$, $p_{f_{\hat{x}_{k+1}^-}}$ is the Taylor polynomial of $f_{\hat{x}_{k+1}^-}$. The function $f_{\hat{x}_{k+1}^-}(\hat{\mathbf{X}}^a_{k},\hat{P}_{x_k}^a)$ can be given as
\begin{align}
    f_{\hat{x}_{k+1}^-}&(\hat{\mathbf{X}}^a_{k},\hat{P}_{x_k}^a)=f(\hat{\mathbf{X}}^a_{k},\mathbf{U}^a_{k}) + \nonumber\\
    &\sum^{n_x}_{i = 0} \omega^{[i]}f(\hat{\mathbf{X}}^a_{k}+\sqrt{(n_x+\lambda)}\sqrt{ \hat{P}_{x_k}^a}^{[i]},\mathbf{U}^a_{k})+\nonumber\\
    &\sum^{2n_x}_{i = n_x} \omega^{[i]}f(\hat{\mathbf{X}}^a_{k}-\sqrt{(n_x+\lambda)}\sqrt{ \hat{P}_{x_k}^a}^{[i]},\mathbf{U}^a_{k})
\end{align}
where $w^{[j]}$ and $\lambda$ are the parameters used in UT.
\end{theorem}
\begin{proof}
Under Assumption \ref{UKF_center}, $\hat{P}_{x_k}^a$ should be regarded as a matrix of values rather than a matrix of sets. By replacing the Taylor model over-approximation of $\hat{\mathbf{P}}_{x_k}^a$ with $\hat{P}_{x_k}^a$, Proposition \ref{estimated states set Over-approximation} can be simplified to \eqref{simplified estimated state set}.
\end{proof}

In this subsection, we introduce the reachability analysis of ASR set. In the next subsection, we overview the algorithms used for reachability analysis.

\subsection{Algorithm overview}
The previous analysis is integrated as an algorithm. As shown in Fig. \ref{approximate_algorithm}, the algorithm consists of the following four steps, where steps 2 and 3 will be iterated.


\begin{figure}
  \centering
  \includegraphics[width=0.95\linewidth]{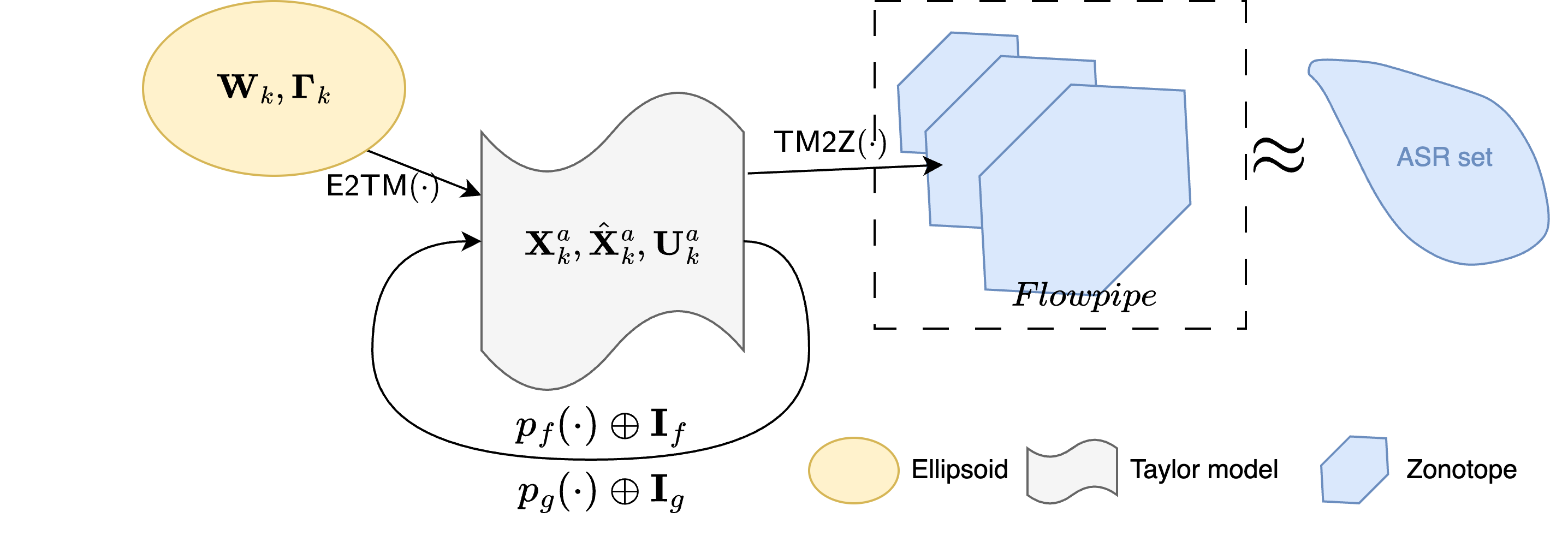}
  \caption{Algorithm overview: After initialization, we iterate and approximate the state set, and take the union of the approximated state sets as the approximation of the ASR set.}
  \label{approximate_algorithm}
\end{figure}

\subsubsection{\textbf{Step 1: Initialization}}\label{Initialization} 
Assume that the UKF has stabilized before reachability analysis. Using the parameters from UKF at $t=t_0$, i.e., $\hat{x}_0$ and $\hat{P}_{x_0}$, we can initialize the state set as follows:
\begin{align}
    [\hat{x}_{1}^-, \hat{P}_{x_1}^-, \sim] &= \mathsf{UT}_{f}(\hat{x}_0, \hat{P}_{x_0})\\
    K^a_0&=K_0\\
    \hat{\mathbf{X}}^a_1 &\approx \hat{x}_{1}^- \oplus K_{1}\cdot\mathsf{E2TM}(\mathbf {\Gamma}_{1})\\
    \mathbf{X}^a_1&\approx f(x_0, u_0)\oplus \mathsf{E2TM}(\mathbf{W}_0)
\end{align}
where $\mathsf{E2TM}$ represents the algorithm that over-approximates an ellipsoid to a Taylor model (see \ref{E2TM}).


\subsubsection{\textbf{Step 2: UKF update}}\label{UKF_Update} We update the UKF parameters as follows:
\begin{align}
  \hat x^a_{k-1} &\!=\! \mathsf{center}(\hat{\mathbf{X}}_{k-1}^a) \\
  u^a_{k} &\!=\! g(\hat{x}^a_{k-1})\\
  [\hat{x}^-_{k}\!,\! \hat P_{x_{k}}^-\!,\! \hat P_{y_{k}}\!,\! \hat P_{xy_{k}}] &\!=\! \mathsf{UT}_{u_k,P_{w},P_{v}}(\hat{x}_{k}, \hat P_{x_{k}})\\
  K^a_{k} &\!=\! \hat P^a_{xz_{k}}(\hat{P}^a_{z_{k}})^{-1}\\
  \hat P^a_{x_{k}} &\!=\! \hat P^{a-}_{x_{k}}\!-\!K^a_{k} \hat P^a_{y_{k}} [K^a_{k}]^T
\end{align}
and then we update the estimated state set under attacks as follows:
\begin{align}
    \hat{\mathbf{X}}_{k+1}^{a-}&\approx p_{f_{\hat{x}_{k+1}^-}}(\hat{\mathbf{X}}^a_{k},\hat{P}_{x_k}^a)\oplus \mathbf{I}_{f_{\hat{x}_{k+1}^-}}\\
    \hat{\mathbf{X}}_{k+1}^{a}&\approx \hat{\mathbf{X}}_{k+1}^{a-} \oplus K^a_{k+1} \cdot \mathsf{E2TM}(\mathbf{\Gamma}_{k+1}).
\end{align}
\subsubsection{\textbf{Step 3: System update}}\label{System update} As described in Theorem (\ref{state_set}), we can over-approximate the ASR set by iterating the state sets as follows:
\begin{align}
    \mathbf{U}^a_{k}&\approx p_g(\hat{\mathbf{X}}_{k-1}^a)\oplus\mathbf{I}_g\\
    \mathbf{X}_{k+1}^a&\approx p_f(\mathbf{X}^a_{k}, \mathbf{U}^a_k)\oplus \mathbf{I}_f \oplus \mathsf{E2TM}(\mathbf{W}_k)
  \end{align}
\subsubsection{\textbf{Step 4: Integration}}\label{Integration} This step is the visualization of the Taylor Model and intersection detection difficulties. So, we over-approximate the results of Step 3 to Zonotopes. We define a list $flowpipe$ that integrates the iteration results of Step 3. as follows:
\begin{equation}\label{flowpipe}
  Flowpipe = \bigcup_{k=1}^{n} \mathsf{TM2Z}(\mathbf{X}^a_k) 
\end{equation}
where $\mathsf{TM2E}$ is the algorithm for over-approximating a Taylor model to a Zonotope \cite{chen2012taylor}. Then the ASR set can be approximated by the $flowpipe$ as
\begin{align}
    \mathcal R^a_{f}(\mathbf{X}_0, &\mathbf{T}^a_{1:n})\approx Flowpipe
\end{align}
Alg. \ref{ASR_approximation} provides the pseudocode for approximating the ASR set. Obtaining the Taylor polynomial is considered the most time-consuming operation in this algorithm.

\begin{algorithm} \SetKwData{Left}{left}\SetKwData{This}{this}\SetKwData{Up}{up} \SetKwFunction{Union}{Union}\SetKwInOut{Input}{input}\SetKwInOut{Output}{output}\SetKw{Push}{push}\SetKwIF{If}{ElseIf}{Else}{if}{then}{else if}{else}{end if}\SetKw{return}{return}

  \caption{SRA algorithm}\label{ASR_approximation} 
	\Input{
   Initialization parameters: $\hat{x}_0$, $\hat{P}_{x_0}$, $K_0$, $P_{w_{0:n}}$;\\
   Attacked System: \eqref{attacked system} and \eqref{alert};\\
   Taylor order: $\kappa$.\\
   }
	\Output{$Flowpipe$}
	 \BlankLine 

  $Flowpipe \leftarrow \emptyset$\;
  $x_0=\hat{x}_0$\;
  \emph{Treat the first step with initialization (see \ref{Initialization})}\;

	 \For{$k=1, \dots, n$}{
    
    $p_f(x, u)+r_f(x, u) \  \leftarrow$ Compute the Taylor expansion of $f(x, u)$ with respect to the variables $x$ and $u$ (see \eqref{Taylor expansion})\;

 $p_g(x)+r_g(x) \  \leftarrow$ Compute the Taylor expansion of $g(x)$ with respect to the variables $x$ (see \eqref{Taylor expansion})\;

    $\hat{\mathbf{X}}_{k+1}^{a}\ \leftarrow$ Update the UKF parameters and estimated state set (see Step \ref{UKF_Update}) \;

    $\mathbf{U}^a_{k+1} \leftarrow$ Approximate the set of the controller output(see Sec. \ref{System update}) \;
    
    $\mathbf{X}_{k+1}^a \leftarrow$ Approximate the state set(see Sec. \ref{System update}) \;
  
    $Flowpipe \leftarrow Flowpipe \cup \mathsf{TM2Z}(\mathbf{X}^a_{k+1})$ \;
    }
    \return{$Flowpipe$}
\end{algorithm}



\section{Construction method and RR Metric}\label{Construction method and Metric}
This section presents our construction method for the risk field and the proposed metric for the risk assessment.

\subsection{Scenario analysis}
Scenario analysis is crucial in a risk assessment\cite{cozzani2005assessment, li2022review, batrouni2018scenario}. A scenario is a hypothetical situation \cite{batrouni2018scenario} consisting of events that may lead to various future consequences. It can be represented by a list of events, such as
\begin{equation}
  Scenario = \{Event_i\},\ i=1, 2,\dots, n_{Event}.
\end{equation}

These events may have undesirable consequences, such as damage to the system and other objects in the scenario\cite{9090897}. It is necessary to identify the critical regions and estimate their risk of these events that an attacked system may trigger in a given scenario. In the case of a system under stealth attacks, accidents occur when its state falls into a critical region. For instance, if the turning speed of a vehicle falls within a high-value interval, the probability of a crash occurring is increased. Therefore, we take "the system state being in critical regions" as an origin event, which is the start of our scenario analysis. These events can result in different losses based on alternative actions of the system and other objects in the scenario. In this article, we describe the events in the scenario as shown in Fig. \ref{event}. Considering the correlation among the events, there are multiple types of structures such as parallel, tree, and net \cite{li2022review}.

Without loss of generality, we present the elicit process as follows:
\begin{equation}
  [Region_i, Risk_i] = \mathsf{Derive}_{\gamma_i}(Event_i)
\end{equation}
where $Region_i$ is the region of critical states, $Risk_i$ is the risk value of the $Event_i$, and $\gamma_i$ represents the probability that $Event_i$ occurs when the system state is in $Region_i$. The possible region of some events is infinite, so we use probability $\gamma$ to limit it. Meanwhile, the detailed steps of $ \mathsf{Derive}_{\gamma}$ can be written as follows.
\begin{enumerate}
  \item Define the scope of the scenario analysis, including the state space and time interval.
  \item Based on expert knowledge or empirical data, derive the critical regions $Region_i$ with probability $\gamma$. \label{step2}
  \item Take these critical regions as the original events of the scenario. Analyze the possible consequences of all the events. \label{step3}
  \item Determine the probability and severity of these consequences. Comprehensively evaluate $risk_i$ using the expected value of loss. \label{step4}
  \item Determine the correlation between events and describe it using a correlation matrix. \label{step5}
\end{enumerate}

In Step \ref{step2}), there are two ways to derive critical regions: analyzing which regions may trigger accidents based on the system model and deducing the source of accidents based on historical records. The derivation of critical regions needs to follow several criteria. Each critical region should correspond to only one event and be derived as accurately as possible without leaving room for error. For events with uncertain critical regions, the critical regions need to be predicted based on the given probability $\gamma$. For dynamic events, i.e., the events result in variables over time, their critical area is deformable and moveable.

Steps \ref{step3}), \ref{step4}), and \ref{step5}) fall outside the purview of this paper. There exist multiple efficient tools to accomplish these steps, including Event Tree Analysis (ETA) \cite{ferdous2009handling}, Hazard and Operability (HAZOP)\cite{dunjo2010hazard}, and Failure Modes and Effects Analysis (FMEA)\cite{chiozza2009fmea}. These tools offer qualitative and quantitative techniques that can assist analysts in identifying potential events and their associated risks in a given scenario\cite{sun2023contradictions}. To ensure comprehensive coverage and avoid omissions, we recommend employing a combination of top-down and bottom-up approaches.

\begin{figure}
  \centering
  \includegraphics[width=0.85\linewidth]{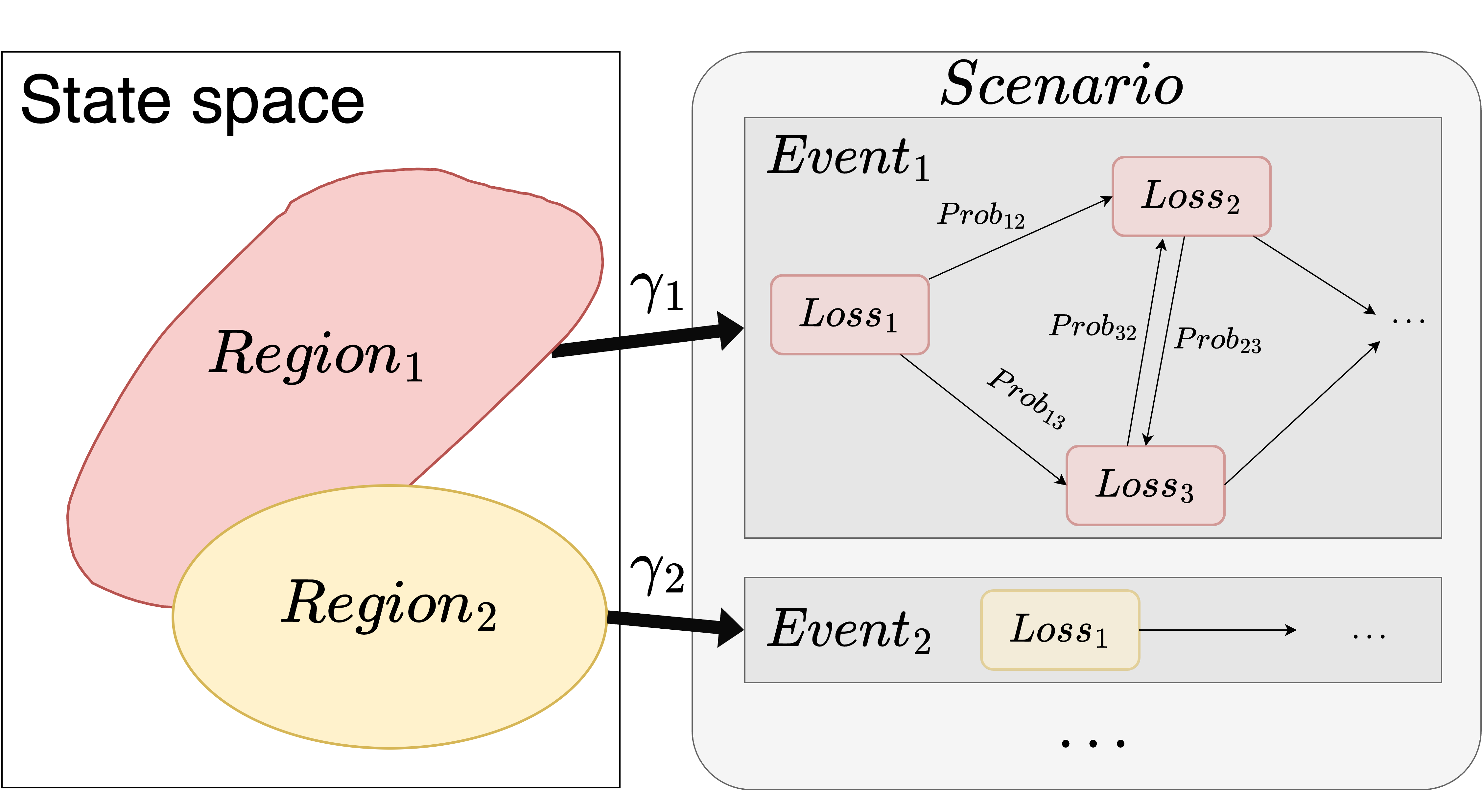}
  \caption{An event occurrence with a given probability $\gamma$ when the system state is in the critical region.}
  \label{event}
\end{figure}

\subsection{Risk field}\label{Risk field}
The risk field is a novel concept defined in this paper, which is used to describe the risk gradient in the state space. First, based on the result of the above scenario analysis, we denote the risk set as
  \begin{equation}
    \begin{split}
      Riskset_i:=\{ \mathbf{S}_i, R_i\}
    \end{split}
  \end{equation}
where
\begin{align}
    \mathbf{S}_i&=\mathsf{Zonotope}(Region_i)\\ 
    R_i&= \gamma_i Risk_i.
\end{align}
where $\mathsf{Zonotope}$ is the algorithm of  Zonotope approximation. $Region_i$ may be a set without a standard shape, and we approximate it to simplify the computation.

Similar to the risk dilution in classical risk analysis, we can perform a dilution operation on the risk set. Obviously, for an event in a scenario, its risk value will be diluted accordingly if its critical region becomes larger. Because the risk of the event will remain invariant, and the probability $\gamma$ will decrease. Denote a sequence of diluted risk sets as 
\begin{equation}
  \mathbf{Risksets}_i=\{\lambda^j Riskset_i\},\ j=0,1,\dots,l
\end{equation}
where $\lambda^j Riskset_i$ represents $j$-th dilution operation for the risk sets. The process of risk set dilution is
\begin{equation}
  \lambda^j Riskset_i = \{ \lambda^j \mathbf{S}_i, \frac{1}{\lambda^{j\cdot n_s}} R_i\}
\end{equation}
where $\lambda>1$ is the dilution factor and $n_s$ is the dimension of $\mathbf{S}_i$. The operation of $\lambda^j \mathbf{S}_i$ refers to \eqref{Zonotope_affine}. A diluted risk set can serve as an early warning, as it has a larger critical area. On the other hand, with the number of risk set dilutions increasing, more risk gradients can enhance the precision of our risk quantification. Thus, we can construct a risk field with the diluted risk sets.

Finally, the risk field constructed for a scenario consists of multiple diluted risk set sequences, which can be written as
\begin{align}
   RiskField = \{\mathbf{Risksets}_i\},\ i=1, 2,\dots, m.
\end{align}

\subsection{RR Metric}\label{metric}
This article attempts to incorporate the dynamic behavior of the attacked system and the potential consequences of the attacks into the risk assessment. Therefore, we present a metric based on the $Flowpipe$ and $RiskField$, which we will introduce in this subsection. 

Stealth attacks can manipulate the state estimation without being detected. If the ASR set intersects with critical regions, as shown in Fig. \ref{paperflow2}, this means that stealth attacks can induce the system state into these regions. The corresponding events will occur due to the effect of the attack. Based on this, we propose a metric to quantify the attack system's risk by detecting the intersection of Flowpipe and Riskset.

\begin{figure}
  \centering
  \includegraphics[width=0.7\linewidth]{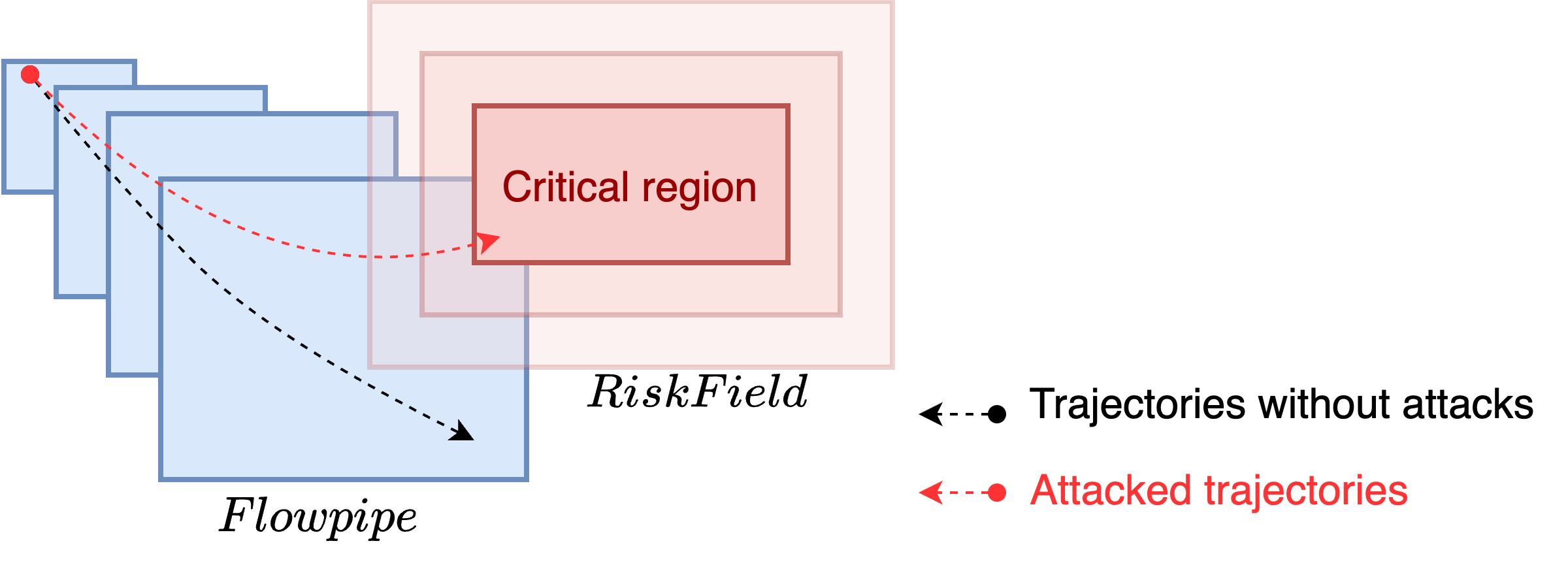}
  \caption{RR metric: Quantifying the risk value based on ASR set approximation and the risk field.}
  \label{paperflow2}
\end{figure}
Such a metric has the following advantages: Reachability analysis describes the future behavior of the attacked system, and the diluted risk set also describes the risk gradient. Therefore, we can assess the risk for a future period. On the other hand, considering the movement of objects in the scene, our assessment is dynamic. Compared to metrics such as the distance from the reachable set to critical states, the RR metric can integrate event relevance, which is conducive to risk assessment in complex scenarios.

Assume that $C\in \mathbb{R}^{m\times m}$ is the correlation matrix, and its element $C^{[i][j]}\in [-1.0, 1.0]$ is the correlation coefficient between $Event_i$ and $Event_j$. We can calculate the total risk of the scenario as follows:
\begin{align}
  Risk\ value = \beta \cdot \zeta C \zeta^T 
\end{align}
where $\beta$ is the probability defined in \eqref{ASR set}, and $\zeta=[\zeta^{[1]}, \zeta^{[2]}, \dots, \zeta^{[m]}]$ is the risk vector of all events. The component $\zeta^{[i]}$ is the risk value of $Event_i$, which is given by
\begin{equation}
  \zeta^{[i]} = \max({\varepsilon_{j}}),\ j=0,1,\dots,l
\end{equation}
where $\max$ is a maximum function, and $\varepsilon_{j}$ is an intermediate variable calculated by 
\begin{equation}
  \varepsilon_{j}=\begin{cases}
    \frac{1}{\lambda^{j\cdot n_s}} R_i& \text{ if } \lambda^j \mathbf{S}_i \cap Flowpipe \ne \emptyset  \\
    0 & \text{ if } \lambda^j \mathbf{S}_i \cap Flowpipe = \emptyset
  \end{cases}
\end{equation}
We introduce a lemma from Theorem 5.2 in \cite{guibas2003zonotopes}, used for collision detection between two Zonotopes.
\begin{lemma}[collision detection]
   For any two zonotopes with $n$ generators in total, we can decide whether they intersect in $\mathcal{O}(n\log^2n)$ time.
\end{lemma}

Based on this lemma, we can check the intersection between $Flowpipe$ and $\lambda^j\mathbf{S}_i$ in $\mathcal{O}(n\log^2n)$ time, where $Flowpipe$ consists of multiple Zonotopes. The pseudocode of the metric is shown in Alg. \ref{Metric}.

\begin{algorithm} \SetKwData{Left}{left}\SetKwData{This}{this}\SetKwData{Up}{up} \SetKwFunction{Union}{Union}\SetKwInOut{Input}{input}\SetKwInOut{Output}{output}\SetKw{Push}{push}\SetKwIF{If}{ElseIf}{Else}{if}{then}{else if}{else}{end if}

  \caption{RR metric}\label{Metric} 
	\Input{$Flowpipe$, $RiskField$, $\lambda$, $C$}

	\Output{$Risk\ value$}
	 \BlankLine 
   $Risk\ value\leftarrow 0$ \;

   \For{$i=1, \dots, m$}{
      \For{$j=1, \dots, l$}{
        $\varepsilon_{j} \leftarrow 0$ \;
        \If{$\lambda^j \mathbf{S}_i \cap Flowpipe= \emptyset$}{
          $\varepsilon_{j} \leftarrow \frac{1}{\lambda^{j\cdot n_s}}$ \;
        }
      }
      $\zeta^{[i]}\leftarrow \max (\varepsilon_{0}, \dots, \varepsilon_{l})$
   }
   $Risk\ value\leftarrow \beta \cdot \zeta C \zeta^T $
\end{algorithm}



\section{Illustrative case}\label{Illustrative case}
We illustrate our assessment framework with a case from an automated warehouse.

\subsection{Scenario Analysis}
\begin{figure*}[!t]
  \centering
  \subfloat[\label{safety_scene1}]{\includegraphics[width=0.32\textwidth]{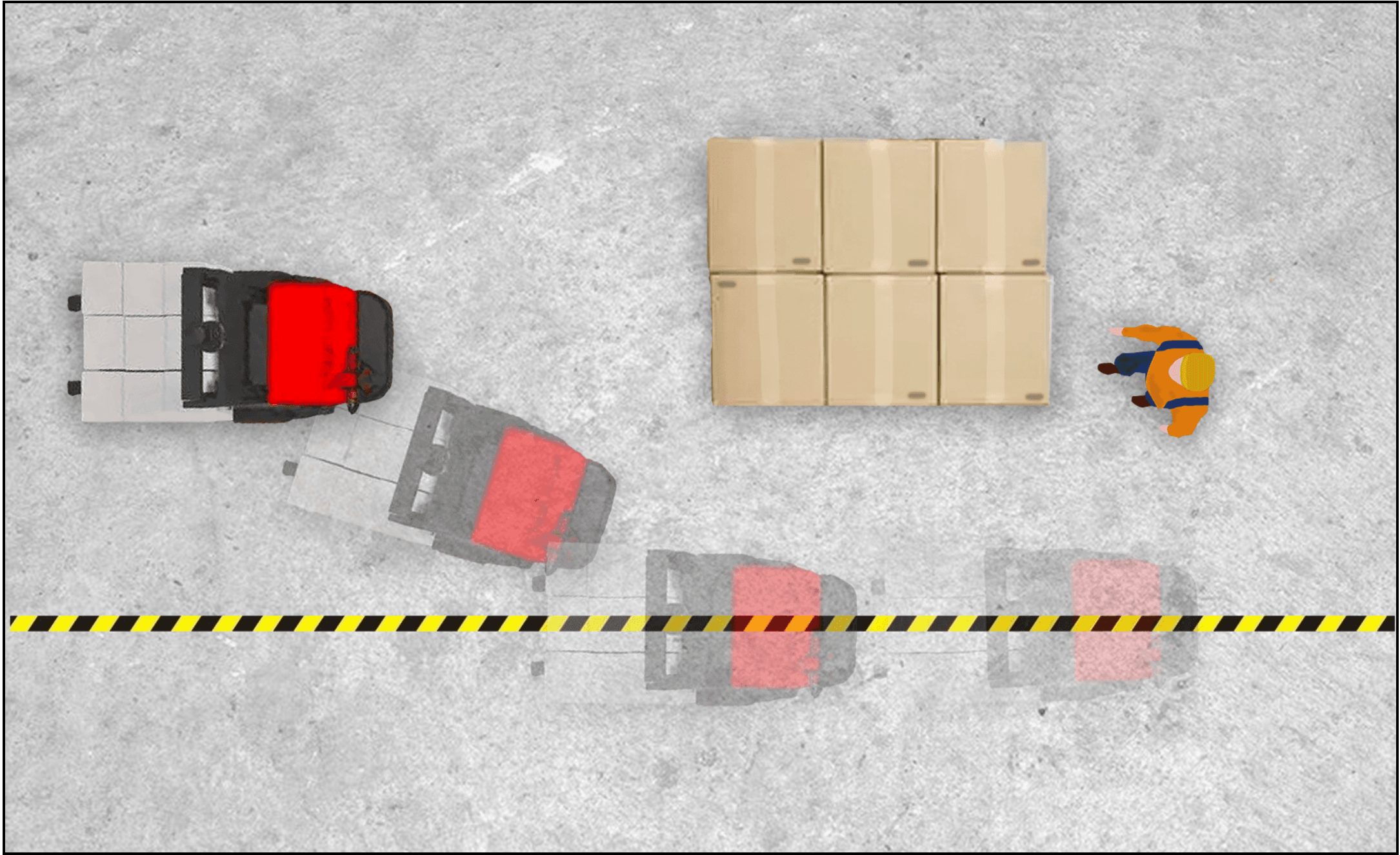}}\hspace{1mm}
  \subfloat[\label{safety_scene2}]{\includegraphics[width=0.32\textwidth]{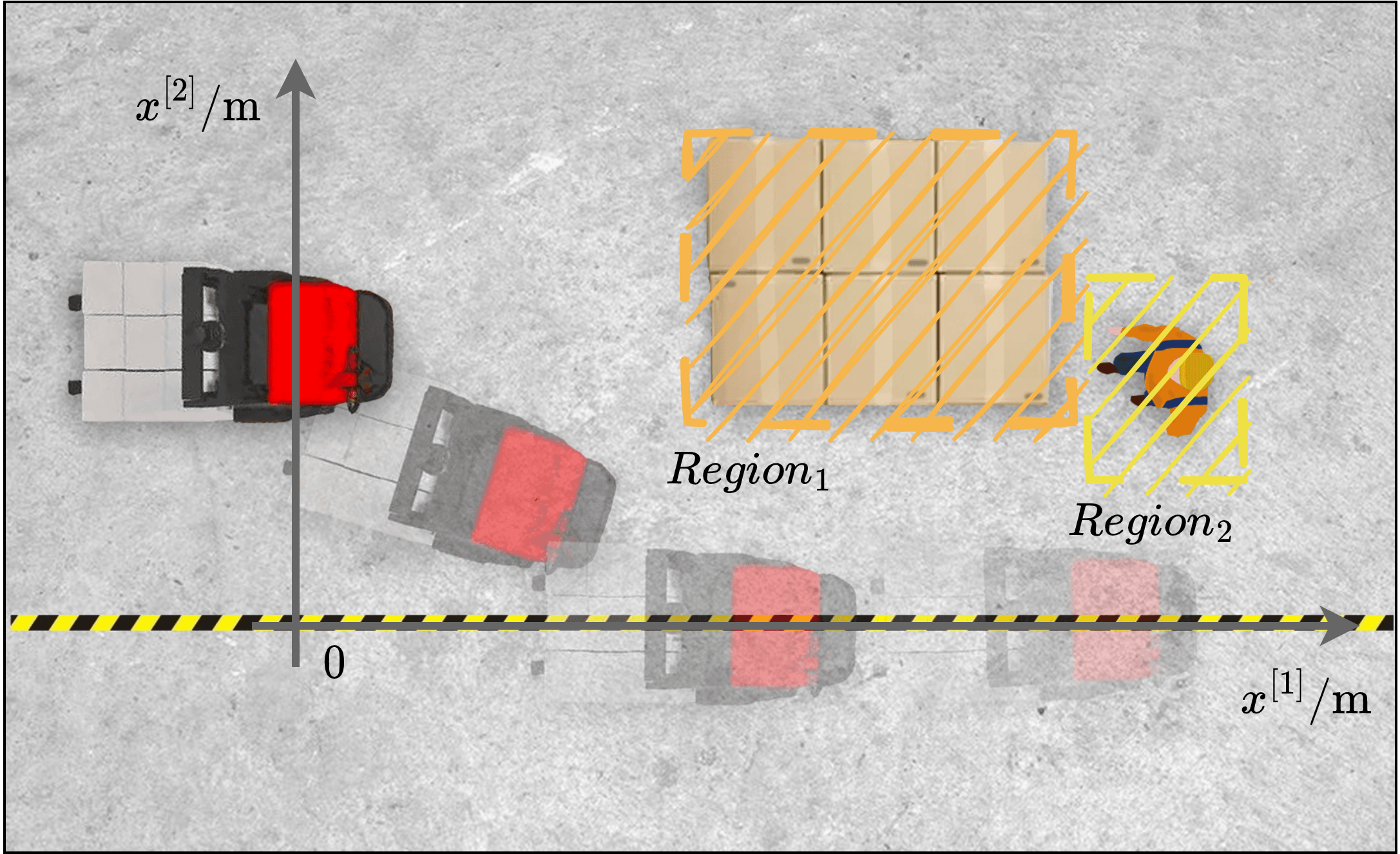}}\hspace{1mm}
  \subfloat[\label{safety_scene3}]{\includegraphics[width=0.32\textwidth]{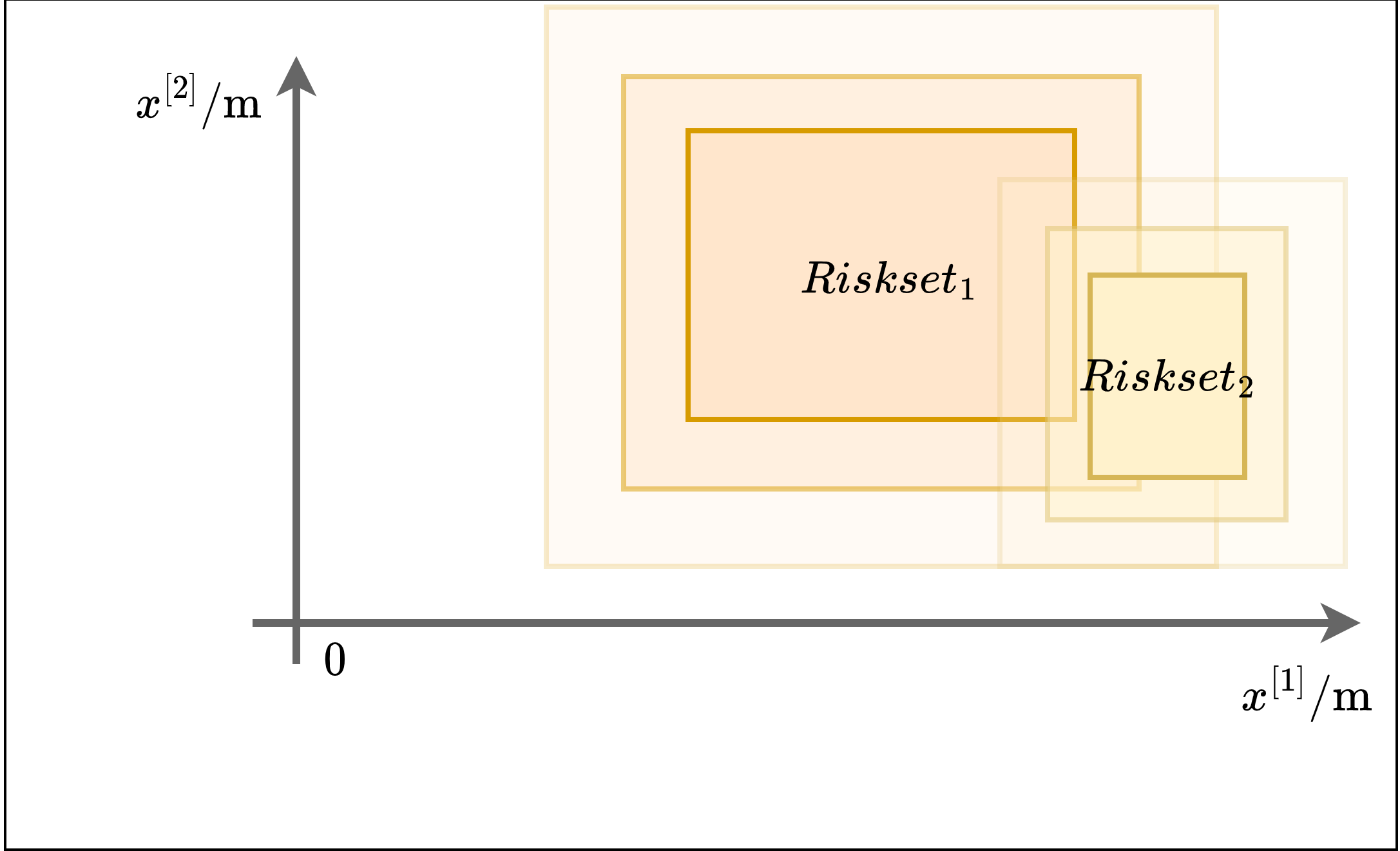}}
  \caption{A case demonstration: Constructing a risk field for the scenario of an unmanned forklift turning process.}
\end{figure*}

As shown in Fig. \ref{safety_scene1}, we consider a scenario involving an unmanned forklift, a worker, and a set of goods as an obstacle. The forklift is equipped with an autonomous navigation function, and it is transporting goods in the warehouse. The obstacle is fixed on the forklift's original path while the worker checks the goods and moves randomly due to his work requirements. Under normal conditions, the forklift detects the obstacle before it and then turns to a preset safe path (yellow striped line) to avoid it. However, the forklift's control is subject to noise and potential attacks. These factors may interfere with the forklift's navigation function and cause deviations from the planned path. As a result, two events may happen while the forklift is changing lanes: 
  \begin{enumerate}
    \item The forklift hits the obstacle and damages the goods or the forklift itself; 
    \item The forklift strikes the worker and causes injury or death. 
  \end{enumerate}

In this scenario, the forklift is chosen as the analyzed system. A coordinate system is established with the origin at the lower-left position of the scenario. The system variables $(x^{[1]}, x^{[2]})$ are defined as the midpoint between the forklift's front wheels, representing the horizontal and vertical distances in meters. Based on experience and knowledge, critical regions for potential collisions with the obstacle and the worker are identified. These regions are illustrated in Fig. \ref{safety_scene2}. If the forklift turns and $(x^{[1]}, x^{[2]})$ falls into these critical regions, there may be a collision with the obstacle or the worker.

To measure the risk values of these events, the ETA is used to analyze the possible consequences of an event. The probabilities and losses of these consequences can be obtained with statistical data and asset inventory, respectively. Assume that the risks of the two events are estimated as $2000$ and $5000$. Then we can formulate two risk sets 
\begin{align}
Riskset_1&=\{ [9.0, 22.0]\times [4.0, 6.5], 2000.0\}\\
Riskset_2&=\{ [24.0, 26.0]\times [3.0, 5.0], 5000.0\}
\end{align}
with $C=\mathsf{diag}([1.0\ 1.0])$. Finally, the risk field can be constructed as shown in Fig. \ref{safety_scene3}

\subsection{System model}
We use a two-wheeled model to describe the dynamic of the unmanned forklift, which is from the survey\cite{paden2016survey}. Its continuous dynamic is given by
\begin{align}
    {x}^{[1]}_{k+1} &= {x}^{[1]}_k +  \delta t \cdot v_0 \cdot \cos (x^{[3]}_k+u^{[1]}_k) + w^{[1]}_k \\
    {x}^{[2]}_{k+1} &= {x}^{[2]}_k + \delta t \cdot v_0 \cdot \sin (x^{[3]}_k+u^{[1]}_k) + w^{[2]}_k \\
    {x}^{[3]}_{k+1} &= {x}^{[3]}_k + \frac{\delta t \cdot v_0}{L}\sin(u^{[1]}_k) 
\end{align}
where $x_k=[x^{[1]}_k\  x^{[2]}_k\  x^{[3]}_k]^T$ are the system state variables, representing the horizontal, vertical positions in meters, and body angle in radians, respectively; $w_k=[w^{[1]}_k\ w^{[2]}_k]^T$ are the process noise with 
$P_{w}= \mathsf{diag}([0.2\ 0.2])$; $u^{[1]}_k$ is the front wheel angle in radians as the control input, $t$ is time in seconds, $L=\SI{2.0}{m}$ is the length of the forklift, $v_0=\SI{5.0}{m/s}$ is the speed constant, and the discrete-time step is $\delta t= \SI{0.1}{s}$. The sensors are given by
\begin{align}
    y_k^{[1]} &= x^{[2]}_k + v^{[1]}_{k} \\
    y_k^{[2]} &= x^{[3]}_k + v^{[2]}_{k}
\end{align}
where $v_k=[v^{[1]}_{k}\ v^{[2]}_{k}]^T$ are the Gaussian measurement noise with $P_{v}=\mathsf{diag}([0.2\ 0.2])$, and we use the Stanley controller to guide the forklift:
\begin{align}
    u^{[1]}_k &= \hat{x}^{[3]}_k+\arctan(\dfrac{k_g \hat{x}^{[2]}_k}{k_s+v_0})
\end{align}
where $k_g=0.1$ and $k_s=0.05$ are the control parameters.


\subsection{Experimental Results}
This section presents experiments designed to illustrate our risk assessment framework. These experiments encompass the over-approximation of the ASR set and the risk assessment of the forklift's lane-changing process.

We first illustrate the over-approximation capability of our analytical algorithm. We approximate the ASR set for a two-wheel system model of the forklift over the time interval $\mathbf{T}^a_{1:n}=[0.0, 1.0]$. The UKF has stabilized before $t= \SI{0.0}{s}$, without any attacks. The initial parameters of the reachability analysis are set as $\hat{x}_{0}^{a}= [0.0\ 5.0\ 0.0]^T$ and $\hat{P}_{x_{0}}^{a}=\mathsf{diag}([0.03\ 0.03\ 0.001])$. For a consistent analysis, the probabilities $\alpha$, $\beta$, and $\gamma$ are all set to $0.95$. As shown by the light red region in Fig. \ref{ASR_set_approximate_T}, we visualize the result of the over-approximated system's reachability, referred to as the Flowpipe, during $\mathbf{T}^a_{1:n}=[0.0, 1.0]$. Concurrently, we utilize a multitude of attacked trajectories to delineate the exact reachability bounds of the system under stealth attacks for comparison. From $t=0.0$, the sensors undergo stealth attacks in 500 Monte-Carlo simulations. We select attack vectors $a_k=[a_k^{[1]}\ a_k^{[2]}]^T$ from a uniform distribution across $[-1, 1]\times [-1, 1]$. Crucially, we ensure these vectors are stealthy by discarding any that would trigger an immediate alarm, replacing them with new random selections. This method generates a sequence of stealth attacks, $a_{1:n}\in\mathbf{\mathcal{A}}_{1:n}$. Attack trajectories are depicted as red lines in Fig. \ref{ASR_set_approximate_T}, with the majority falling within the flowpipe. Some trajectories exceed the bounds due to process noise, as our analysis accounts for bounded noise with a probability $\beta=0.95$. These results validate our algorithm's ability to over-approximate the reachability of a system under attack.

\begin{figure}
  \centering
    \includegraphics[width=0.80\linewidth]{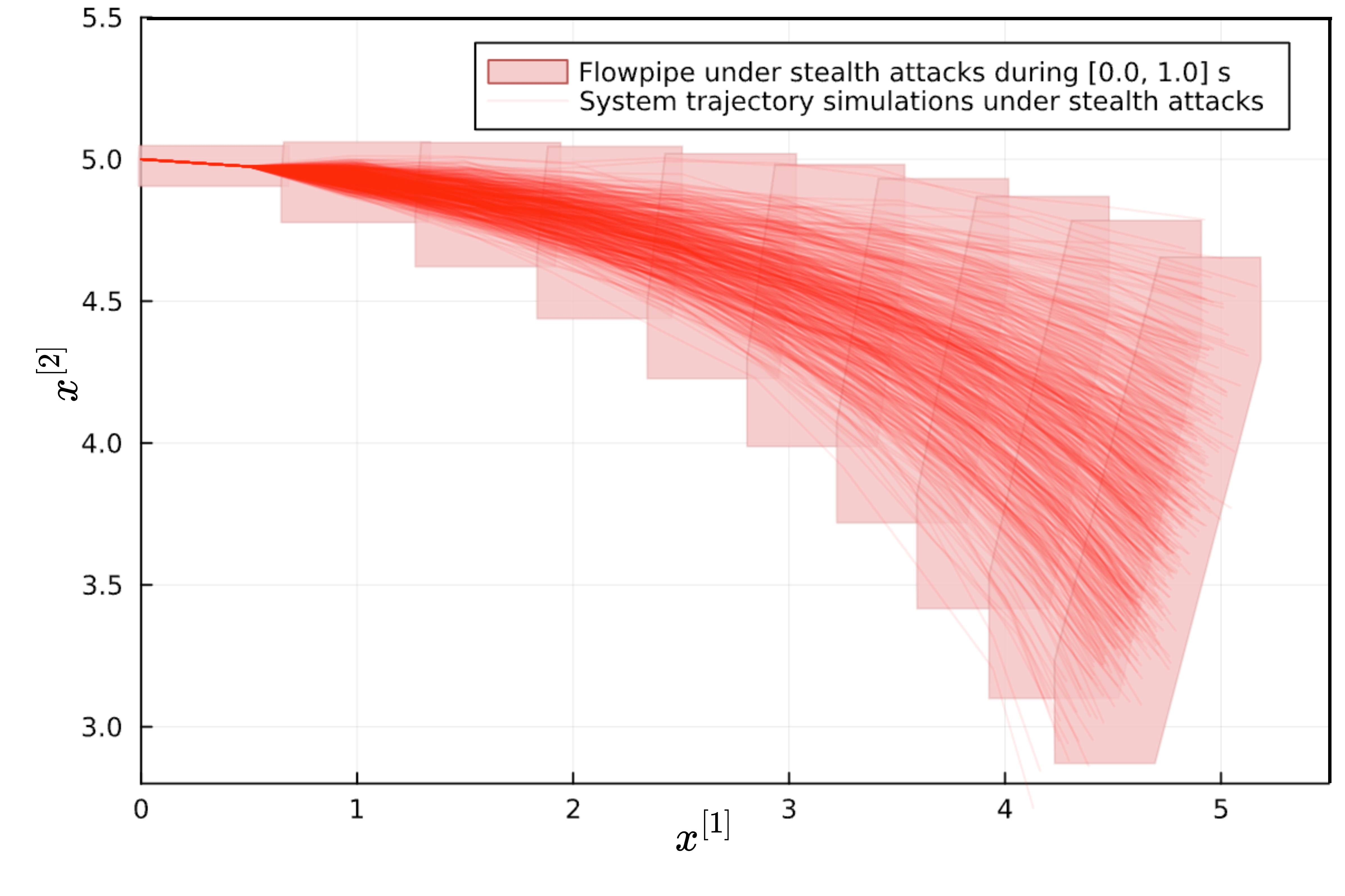}
    \caption{The approximation of ASR set over $\mathbf{T}^a_{1:n}=[0.0, 1.0]$, projected onto the $(x^{[1]}, x^{[2]})$ plane.}
    \label{ASR_set_approximate_T}

    \includegraphics[width=0.80\linewidth]{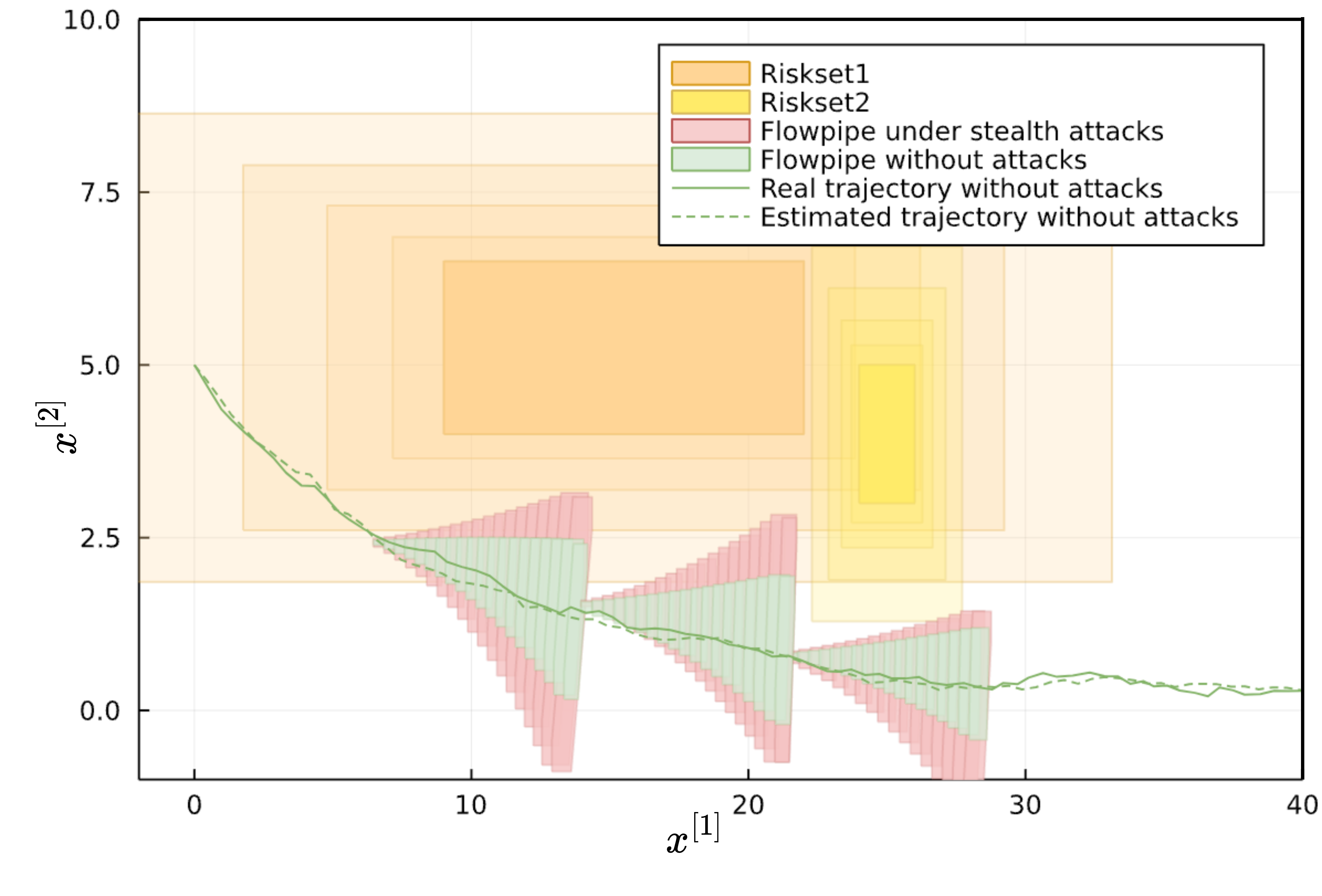}
    \caption{Simulation of forklift path-changing process and risk assessment.}
    \label{Reach_filter_simulation}

    \includegraphics[width=0.80\linewidth]{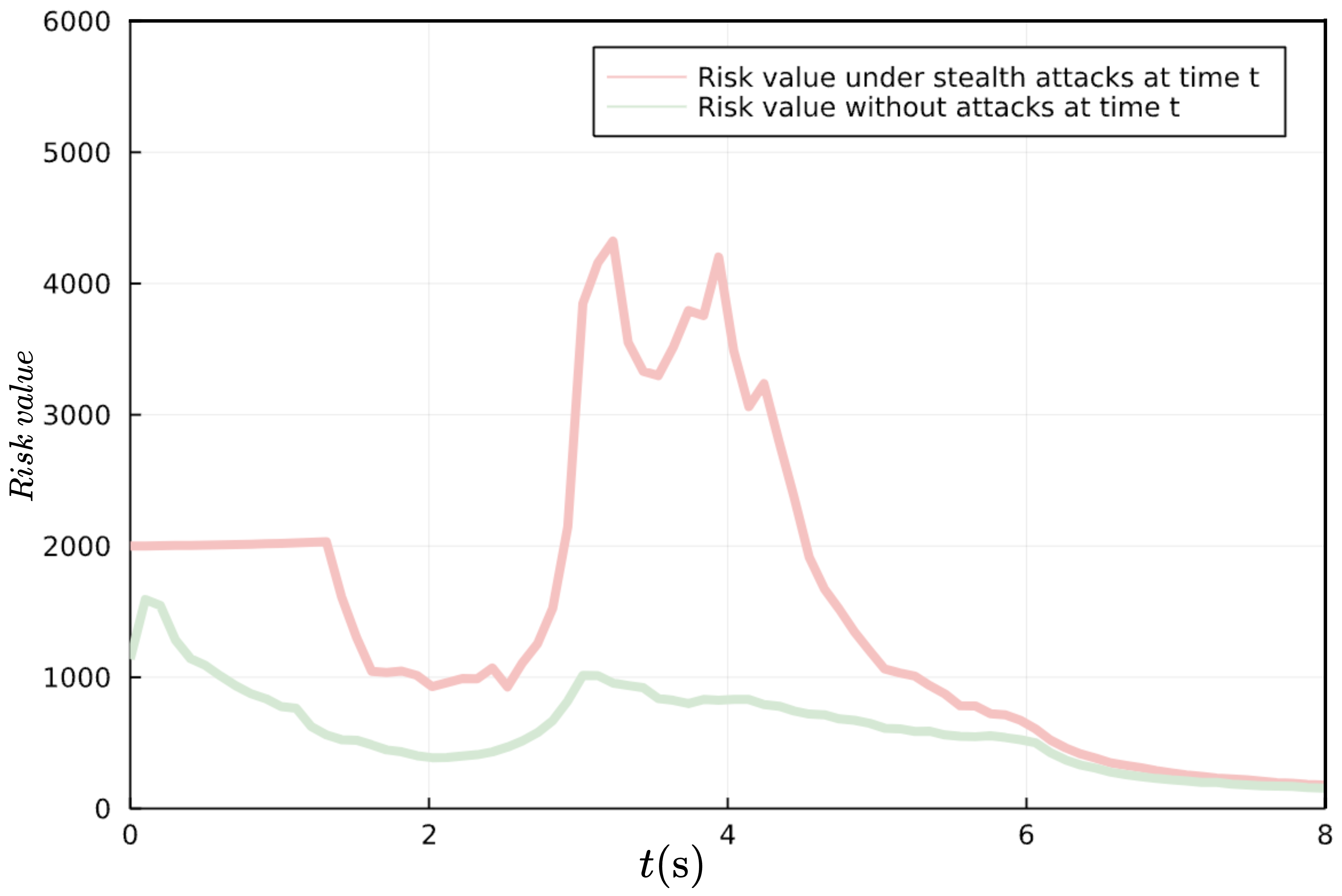}
    \caption{Dynamic risk value of the forklift lane changing process.}
    \label{risk_real_time}

\end{figure}

The effectiveness of the proposed framework is demonstrated by applying it to the forklift turning scenario. First, we simulate the state estimation of the forklift turning process over the interval $t\in [0.0, 10.0]$, as the green line (true trajectory) and the green dotted line (estimated trajectory) shown in Fig. \ref{Reach_filter_simulation}. In the absence of any attack, the estimated state tracks the true state with an acceptable estimation error, and the forklift correctly approaches the preset safe path. Subsequently, we approximate the reachability of the forklift under stealth attacks at each time step. The ASR set is approximated as flowpipes over the interval $[t_k, t_k + 1.5]$, where $t_k = 0.0, 0.1, \dots, 10.0$, as the light red regions shown in Fig. \ref{Reach_filter_simulation}. Meanwhile, to make a comparison, we also compute the flowpipe in the absence of attacks, as shown in the light green regions shown in Fig. \ref{Reach_filter_simulation}. In the case of not considering stealth attacks, $r_k^T \hat{P}_{r_k}^{-1} r_k\leq a$, where the value of $a$ represents the average value obtained during state estimation simulation. Note that, for brevity, there are only three flowpipe comparisons over time intervals $[1.5, 3.0]$, $[3.0, 4.5]$, and $[4.5, 6.0]$ are drawn in the figure. 

The risk field comprises two lists of diluted risk sets, with a dilution factor of $\lambda = 1.05$, corresponding to the two events derived from the scenario analysis. For brevity, the risk field is visually represented by plotting partially diluted risk sets, as illustrated in Fig. \ref{Reach_filter_simulation}. Subsequently, the comparative dynamic risk values, with and without consideration for stealth attacks, are computed by employing the RR metric at each time step. As illustrated in Fig. \ref{risk_real_time}, the presence of two peaks in the red line corresponds to the two events that the attacked forklift may trigger at subsequent times. The different waveforms in the two lines illustrate the additional risk introduced by stealth attacks. These results indicate that the proposed framework can effectively predict the risk of stealth attacks within a given scenario. The evaluated risk value integrates the impacts of system dynamics, noise, stealth attacks, and the scenario itself.

\section{Conclusion}\label{Conclusion}
This article presents a novel risk assessment framework for nonlinear systems under stealth attacks. First, we define the system’s reachability under stealth attacks and propose an algorithm to approximate it. Then, we introduce the concept of a risk field to formally describe the risk distribution of a scenario and present a method to construct the field for a given scenario. Following this, we introduce a metric to quantify the risk value based on the reachability approximation and the risk field. Finally, we demonstrate the risk prediction ability of our framework with a numerical example from an automated warehouse. The dynamic results from our framework reveal the system’s potential vulnerability in an explainable way, providing an early warning for safety controls. Future research will focus on enhancing the efficiency of risk assessment and implementing risk-informed safety controls.

\bibliographystyle{unsrt}
\bibliography{reference.bib}

\begin{thebibliography}{10}

\bibitem{henshaw2018research}
Michael Henshaw.
\newblock Research challenges and transatlantic collaboration on transportation cyber-physical systems.
\newblock {\em Transportation Cyber-Physical Systems}, pages 247--265, 2018.

\bibitem{zhang2015health}
Yin Zhang, Meikang Qiu, Chun-Wei Tsai, Mohammad~Mehedi Hassan, and Atif Alamri.
\newblock Health-cps: Healthcare cyber-physical system assisted by cloud and big data.
\newblock {\em IEEE Systems Journal}, 11(1):88--95, 2015.

\bibitem{wittenberg2016human}
Carsten Wittenberg.
\newblock Human-cps interaction-requirements and human-machine interaction methods for the industry 4.0.
\newblock {\em IFAC-PapersOnLine}, 49(19):420--425, 2016.

\bibitem{farwell2011stuxnet}
James~P Farwell and Rafal Rohozinski.
\newblock Stuxnet and the future of cyber war.
\newblock {\em Survival}, 53(1):23--40, 2011.

\bibitem{choras2016cyber}
Micha{\l} Chora{\'s}, Rafa{\l} Kozik, Adam Flizikowski, Witold Ho{\l}ubowicz, and Rafa{\l} Renk.
\newblock Cyber threats impacting critical infrastructures.
\newblock In {\em Managing the Complexity of Critical Infrastructures}, pages 139--161. Springer, Cham, 2016.

\bibitem{Ding2018A}
Derui Ding, Qing~Long Han, Yang Xiang, Xiaohua Ge, and Xian~Ming Zhang.
\newblock A survey on security control and attack detection for industrial cyber-physical systems.
\newblock {\em Neurocomputing}, 2018.

\bibitem{musleh2019survey}
Ahmed~S Musleh, Guo Chen, and Zhao~Yang Dong.
\newblock A survey on the detection algorithms for false data injection attacks in smart grids.
\newblock {\em IEEE Transactions on Smart Grid}, 11(3):2218--2234, 2019.

\bibitem{cardenas2011attacks}
Alvaro~A C{\'a}rdenas, Saurabh Amin, Zong-Syun Lin, Yu-Lun Huang, Chi-Yen Huang, and Shankar Sastry.
\newblock Attacks against process control systems: risk assessment, detection, and response.
\newblock In {\em Proceedings of the 6th ACM symposium on information, computer and communications security}, pages 355--366, 2011.

\bibitem{khazraei2022resiliency}
Amir Khazraei and Miroslav Pajic.
\newblock Resiliency of nonlinear control systems to stealthy sensor attacks.
\newblock In {\em 2022 IEEE 61st Conference on Decision and Control (CDC)}, pages 7109--7114. IEEE, 2022.

\bibitem{zhang2021stealthy}
Kangkang Zhang, Christodoulos Keliris, Thomas Parisini, and Marios~M Polycarpou.
\newblock Stealthy integrity attacks for a class of nonlinear cyber-physical systems.
\newblock {\em IEEE Transactions on Automatic Control}, 2021.

\bibitem{Yang2022A}
Tianci Yang and Chen Lv.
\newblock A secure sensor fusion framework for connected and automated vehicles under sensor attacks.
\newblock {\em IEEE internet of things journal}, 2022.

\bibitem{sui2020vulnerability}
Tianju Sui, Yilin Mo, Dami{\'a}n Marelli, Ximing Sun, and Minyue Fu.
\newblock The vulnerability of cyber-physical system under stealthy attacks.
\newblock {\em IEEE Transactions on Automatic Control}, 66(2):637--650, 2020.

\bibitem{liu2019joint}
Chensheng Liu, Hao Liang, Tongwen Chen, Jing Wu, and Chengnian Long.
\newblock Joint admittance perturbation and meter protection for mitigating stealthy fdi attacks against power system state estimation.
\newblock {\em IEEE Transactions on Power Systems}, 35(2):1468--1478, 2019.

\bibitem{akametalu2014reachability}
Anayo~K Akametalu, Jaime~F Fisac, Jeremy~H Gillula, Shahab Kaynama, Melanie~N Zeilinger, and Claire~J Tomlin.
\newblock Reachability-based safe learning with gaussian processes.
\newblock In {\em 53rd IEEE Conference on Decision and Control}, pages 1424--1431. IEEE, 2014.

\bibitem{teixeira2019optimal}
Andr{\'e}~MH Teixeira.
\newblock Optimal stealthy attacks on actuators for strictly proper systems.
\newblock In {\em 2019 IEEE 58th Conference on Decision and Control (CDC)}, pages 4385--4390. IEEE, 2019.

\bibitem{murguia2020security}
Carlos Murguia, Iman Shames, Justin Ruths, and Dragan Ne{\v{s}}i{\'c}.
\newblock Security metrics and synthesis of secure control systems.
\newblock {\em Automatica}, 115:108757, 2020.

\bibitem{kwon2017reachability}
Cheolhyeon Kwon and Inseok Hwang.
\newblock Reachability analysis for safety assurance of cyber-physical systems against cyber attacks.
\newblock {\em IEEE Transactions on Automatic Control}, 63(7):2272--2279, 2017.

\bibitem{hwang2023lmi}
Sounghwan Hwang, Minhyun Cho, Sungsoo Kim, and Inseok Hwang.
\newblock An lmi-based risk assessment of leader-follower multi-agent system under stealthy cyberattacks.
\newblock {\em IEEE Control Systems Letters}, 2023.

\bibitem{fan2021improved}
Jianwei Fan, Jun Huang, and Xudong Zhao.
\newblock Improved interval estimation method for cyber-physical systems under stealthy deception attacks.
\newblock {\em IEEE Transactions on Signal and Information Processing over Networks}, 8:1--11, 2021.

\bibitem{zhang2020reachability}
Qirui Zhang, Kun Liu, Zhonghua Pang, Yuanqing Xia, and Tao Liu.
\newblock Reachability analysis of cyber-physical systems under stealthy attacks.
\newblock {\em IEEE Transactions on Cybernetics}, 52(6):4926--4934, 2022.

\bibitem{hashemi2018comparison}
Navid Hashemi, Carlos Murguia, and Justin Ruths.
\newblock A comparison of stealthy sensor attacks on control systems.
\newblock In {\em 2018 Annual American Control Conference (ACC)}, pages 973--979. IEEE, 2018.

\bibitem{liu2021reachability}
Hao Liu, Ben Niu, and Jiahu Qin.
\newblock Reachability analysis for linear discrete-time systems under stealthy cyber attacks.
\newblock {\em IEEE Transactions on Automatic Control}, 66(9):4444--4451, 2021.

\bibitem{wang2014iss}
Xiangke Wang, Jiahu Qin, and Changbin Yu.
\newblock Iss method for coordination control of nonlinear dynamical agents under directed topology.
\newblock {\em IEEE transactions on cybernetics}, 44(10):1832--1845, 2014.

\bibitem{lyu2019safety}
Xiaorong Lyu, Yulong Ding, and Shuang-Hua Yang.
\newblock Safety and security risk assessment in cyber-physical systems.
\newblock {\em IET Cyber-Physical Systems: Theory \& Applications}, 4(3):221--232, 2019.

\bibitem{sarkka2007unscented}
Simo Sarkka.
\newblock On unscented kalman filtering for state estimation of continuous-time nonlinear systems.
\newblock {\em IEEE Transactions on automatic control}, 52(9):1631--1641, 2007.

\bibitem{chambers1976method}
John~M Chambers, Colin~L Mallows, and BW4159820341 Stuck.
\newblock A method for simulating stable random variables.
\newblock {\em Journal of the american statistical association}, 71(354):340--344, 1976.

\bibitem{golub2013matrix}
Gene~H Golub and Charles~F Van~Loan.
\newblock {\em Matrix computations}.
\newblock JHU press, 2013.

\bibitem{wu2023double}
Yifei Wu, Yinrui Ma, and Ye~Jiang.
\newblock Double-state chi-square test based sparse grid quadrature filtering algorithm and its application in integrated navigation.
\newblock {\em IET Control Theory \& Applications}, 17(9):1203--1213, 2023.

\bibitem{julier2004unscented}
Simon~J Julier and Jeffrey~K Uhlmann.
\newblock Unscented filtering and nonlinear estimation.
\newblock {\em Proceedings of the IEEE}, 92(3):401--422, 2004.

\bibitem{chen2012taylor}
Xin Chen, Erika Abraham, and Sriram Sankaranarayanan.
\newblock Taylor model flowpipe construction for non-linear hybrid systems.
\newblock In {\em 2012 IEEE 33rd Real-Time Systems Symposium}, pages 183--192. IEEE, 2012.

\bibitem{moore2009introduction}
Ramon~E Moore, R~Baker Kearfott, and Michael~J Cloud.
\newblock {\em Introduction to interval analysis}.
\newblock SIAM, 2009.

\bibitem{chen2015reachability}
Xin Chen.
\newblock {\em Reachability analysis of non-linear hybrid systems using taylor models}.
\newblock PhD thesis, Fachgruppe Informatik, RWTH Aachen University, 2015.

\bibitem{kopetzki2017methods}
Anna-Kathrin Kopetzki, Bastian Sch{\"u}rmann, and Matthias Althoff.
\newblock Methods for order reduction of zonotopes.
\newblock In {\em 2017 IEEE 56th Annual Conference on Decision and Control (CDC)}, pages 5626--5633. IEEE, 2017.

\bibitem{cozzani2005assessment}
Valerio Cozzani, Gianfilippo Gubinelli, Giacomo Antonioni, Gigliola Spadoni, and Severino Zanelli.
\newblock The assessment of risk caused by domino effect in quantitative area risk analysis.
\newblock {\em Journal of hazardous Materials}, 127(1-3):14--30, 2005.

\bibitem{li2022review}
Hui Li, Zhouyang Ren, Miao Fan, Wenyuan Li, Yan Xu, Yunpeng Jiang, and Weiyi Xia.
\newblock A review of scenario analysis methods in planning and operation of modern power systems: Methodologies, applications, and challenges.
\newblock {\em Electric Power Systems Research}, 205:107722, 2022.

\bibitem{batrouni2018scenario}
Marwan Batrouni, Aur{\'e}lie Bertaux, and Christophe Nicolle.
\newblock Scenario analysis, from bigdata to black swan.
\newblock {\em Computer Science Review}, 28:131--139, 2018.

\bibitem{9090897}
Stefan Riedmaier, Thomas Ponn, Dieter Ludwig, Bernhard Schick, and Frank Diermeyer.
\newblock Survey on scenario-based safety assessment of automated vehicles.
\newblock {\em IEEE Access}, 8:87456--87477, 2020.

\bibitem{ferdous2009handling}
Refaul Ferdous, Faisal Khan, Rehan Sadiq, Paul Amyotte, and Brian Veitch.
\newblock Handling data uncertainties in event tree analysis.
\newblock {\em Process safety and environmental protection}, 87(5):283--292, 2009.

\bibitem{dunjo2010hazard}
Jordi Dunj{\'o}, Vasilis Fthenakis, Juan~A V{\'\i}lchez, and Josep Arnaldos.
\newblock Hazard and operability (hazop) analysis. a literature review.
\newblock {\em Journal of hazardous materials}, 173(1-3):19--32, 2010.

\bibitem{chiozza2009fmea}
Maria~Laura Chiozza and Clemente Ponzetti.
\newblock Fmea: a model for reducing medical errors.
\newblock {\em Clinica chimica acta}, 404(1):75--78, 2009.

\bibitem{sun2023contradictions}
Zhicong Sun, Yulong Ding, Ke~Pei, and Shuang-Hua Yang.
\newblock Contradictions identification of safety and security requirements for industrial cyber-physical systems.
\newblock {\em IEEE Internet of Things Journal}, 2023.

\bibitem{guibas2003zonotopes}
Leonidas~J Guibas, An~Thanh Nguyen, and Li~Zhang.
\newblock Zonotopes as bounding volumes.
\newblock In {\em SODA}, volume~3, pages 803--812. Citeseer, 2003.

\bibitem{paden2016survey}
Brian Paden, Michal {\v{C}}{\'a}p, Sze~Zheng Yong, Dmitry Yershov, and Emilio Frazzoli.
\newblock A survey of motion planning and control techniques for self-driving urban vehicles.
\newblock {\em IEEE Transactions on intelligent vehicles}, 1(1):33--55, 2016.

\end{thebibliography}

\newpage
\begin{IEEEbiography}[{\includegraphics[width=1in,height=1.25in,clip,keepaspectratio]{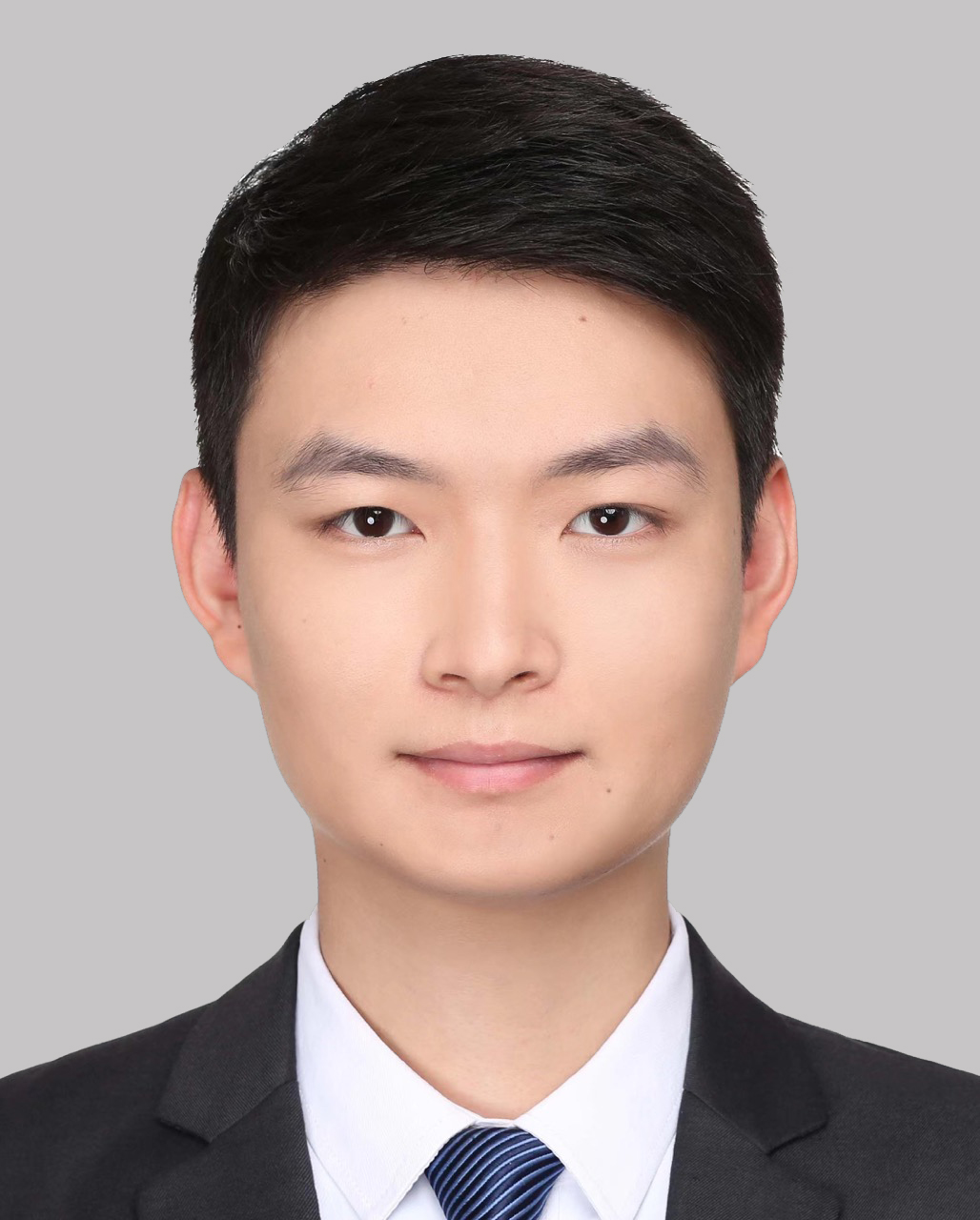}}]{Guang Chen} 
  received his B.E. degree in Materials Processing and Control Engineering from Southern University of Science and Technology, China, in 2019 and his M.Eng degree from the Department of Mechanical and Energy Engineering at Southern University of Science and Technology, China, in 2023. His research focuses on the safety assessment of cyber-physical systems.
\end{IEEEbiography}
\vspace{-10 mm}

\begin{IEEEbiography}[{\includegraphics[width=1in,height=1.25in,clip,keepaspectratio]{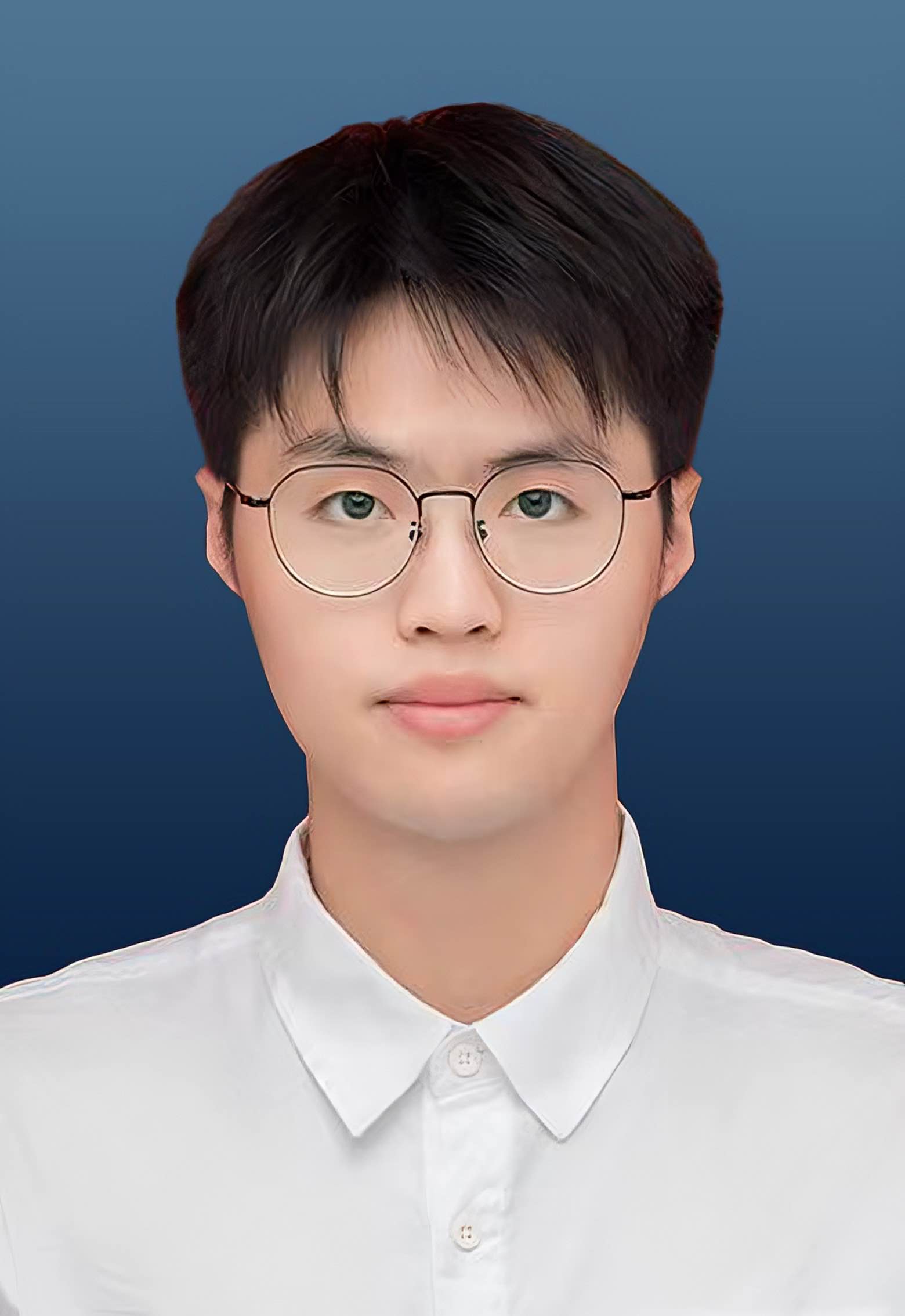}}]{Zhicong Sun}
  received his B.E. degree in Communication Engineering from Harbin Institute of Technology, China, in 2020 and his M.Eng degree from the Department of Computer Science and Engineering at Southern University of Science and Technology, China, in 2023. He is currently a PhD student at Hong Kong Polytechnic University. His research interests include the safety and security of cyber-physical systems.
\end{IEEEbiography}
\vspace{-10 mm}

\begin{IEEEbiography}[{\includegraphics[width=1in,height=1.25in,clip,keepaspectratio]{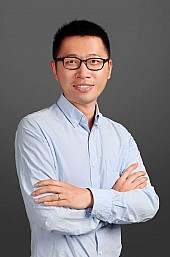}}]{Yulong Ding}
  received his B.A.Sc. and M.A.Sc. degrees in Chemical Engineering from Tsinghua University, China, in 2005 and 2008, respectively. He earned his Ph.D. degree in Chemical Engineering from The University of British Columbia, Canada, in 2012. He is currently a Research Associate Professor at the Department of Computer Science and Engineering at Southern University of Science and Technology. His primary research interests include the industrial Internet of Things and low-power wide area networks (LPWANs).
\end{IEEEbiography}
\vspace{-10 mm}

\begin{IEEEbiography}[{\includegraphics[width=1in,height=1.25in,clip,keepaspectratio]{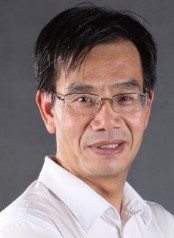}}]{Shuang-Hua Yang}
  (Senior Member, IEEE) received his B.S. degree in Instrument and Automation and his M.S. degree in Process Control from the China University of Petroleum (Huadong), Beijing, China, in 1983 and 1986, respectively. He earned his Ph.D. degree in Intelligent Systems from Zhejiang University, Hangzhou, China, in 1991. He is currently the Director of the Shenzhen Key Laboratory of Safety and Security for Next Generation of Industrial Internet at Southern University of Science and Technology, China, and also serves as the Head of the Department of Computer Science at the University of Reading, UK. His research interests include cyber-physical systems, the Internet of Things, wireless network-based monitoring and control, and safety-critical systems. He is a Fellow of IET and InstMC in the UK and an Associate Editor of IET Cyber-Physical Systems: Theory and Applications.
  \end{IEEEbiography}

  \enlargethispage{-6.5cm}
\end{document}